\documentclass{article}
\usepackage{fullpage,authblk,color,framed}

\pdfoutput=1

\usepackage{paralist}

\usepackage{verbatim}

\usepackage{xcolor}
\usepackage[colorlinks]{hyperref}
\hypersetup{
citecolor=green!50!black,
linkcolor=blue!50!black
}

\def\final{1}

\usepackage{algorithm}
\usepackage{algpseudocode}
\usepackage{amsmath,tikz,graphicx,etoolbox,caption,amssymb}

\colorlet{c1col}{red!80!darkgray}
\colorlet{c2col}{blue!80!darkgray}
\colorlet{c3col}{green!80!darkgray}
\tikzset{
  a/.style={draw=none,inner sep=-.1pt,outer sep=-.	1pt},
  a2/.style={draw=none,inner sep=-1pt,outer sep=-1pt},
  c1t/.style={draw=c1col,line width=5pt, line cap=round, line join=round},
  c2t/.style={draw=c2col,line width=5pt, line cap=round, line join=round},
  c1/.style={draw=c1col,line width=1.4pt, line cap=round},
  c2/.style={draw=c2col,line width=1.4pt, line cap=round},
  c3t/.style={draw=c3col,line width=5pt, line cap=round},
  c3/.style={draw=c3col,line width=1.4pt, line cap=round}
}
\usepackage{fancybox}
\usepackage{refcount}
\usepackage{subcaption,amsthm,url}

\usepackage{wrapfig}

\newtoggle{abs}
\toggletrue{abs}
\togglefalse{abs}
\newtoggle{absfalse}
\iftoggle{abs}{\togglefalse{absfalse}}{\toggletrue{absfalse}}

\newif\ifcomment\commentfalse
\def\commentON{\commenttrue}

\long\outer\def\bcl#1\ecl{{\ifcomment \sloppy  \textcolor{red}{{{#1}}}\fi }}

\long\outer\def\BCL#1\ECL{{\ifcomment \sloppy \par \textcolor{blue}{\#  \dotfill
{\textsc{#1}} \dotfill \#} \par \fi }}

\long\outer\def\BCA#1\ECA{{\ifcomment \sloppy \par \textcolor{green}{\#  \dotfill
{\textsc{#1}} \dotfill \#} \par \fi }}

\commentON

\ifdefined\final
\usepackage{microtype}
\newcommand{\nnote}[1]{}
\newcommand{\anote}[1]{}
\newcommand{\lnote}[1]{}
\newcommand{\fnote}[1]{}

\newcommand{\todo}[1]{}
\newcommand{\modified}[1]{#1}

\else
\newcommand{\nnote}[1]{\begingroup \color{blue!60!black} \em Neil: #1 \endgroup}
\newcommand{\anote}[1]{\begingroup \color{blue!60!black} \em Anke: #1 \endgroup}
\newcommand{\fnote}[1]{\begingroup \color{blue!60!black} \em Frans: #1 \endgroup}
\newcommand{\lnote}[1]{\begingroup \color{blue!60!black} \em Leen: #1 \endgroup}

\newcommand{\modified}[1]{\begingroup \color{purple!60!black} #1\endgroup}
\newcommand{\todo}[1]{\begingroup \color{red!80!black} \em TODO: #1 \endgroup}
\fi

\newtheorem*{remark}{Remark}

\newtheorem{theorem}{Theorem}
\newtheorem{claim}[theorem]{Claim}
\newtheorem*{claim*}{Claim}
\newtheorem{observation}{Observation}

\newtheorem{lemma}[theorem]{Lemma}
\newtheorem{proposition}[theorem]{Proposition}

\newenvironment{proof_claim}{\noindent {\it Proof of Claim: }}{\hspace*{\fill} $\diamond$}

\newcommand{\keywords}[1]{\medskip \noindent \textbf{Keywords: }#1.}
\newcommand{\subjclass}[1]{\medskip \noindent \textbf{MSC: }#1.}

\theoremstyle{definition}
\newtheorem{definition}{Definition}

\title{A Duality Based 2-Approximation Algorithm for Maximum~Agreement Forest
\footnote{This paper is based on the (substantially different) extended abstract~\cite{DBLP:conf/icalp/SchalekampZS16}.}}

\newcommand{\R}{\mathbb{R}}

\DeclareMathOperator{\lca}{lca}

\DeclareMathOperator{\minimize}{minimize}

\usepackage{authblk}

\newcommand{\coneT}{C_{\mathcal{F}}}

\newcommand{\monochrome}{unicolored}
\newcommand{\multicolored}{multicolored}
\newcommand{\bicolored}{bicolored}
\newcommand{\tricolored}{tricolored}
\newcommand{\comp}[1]{#1-compatible}
\newcommand{\feas}[1]{#1-feasible}
\newcommand{\pcs}{root-of-infeasibility}%
\newcommand{\pcsplural}{roots-of-infeasibility}
\newcommand{\compatible}{compatible}
\newcommand{\RBcomp}{$R\cup B$-compatible}
\newcommand{\ptnstart}{\ptn^{(0)}}
\newcommand{\ptnRBcomp}{\ptn^{(1)}}
\newcommand{\ptnsplit}{\ptn^{(2)}}
\newcommand{\ptnfeas}{\ptn^{(3)}}

\newcommand{\ptncur}{\mathcal{Q}}
\newcommand{\pairslist}{{\tt pairslist}}
\newcommand{\redleaf}{\ensuremath{x_R}}
\newcommand{\blueleaf}{\ensuremath{x_B}}
\newcommand{\altredleaf}{\ensuremath{y_R}}

\newcommand{\whiteleaf}{\ensuremath{x_W}}
\newcommand{\nodeinTtwo}[1]{\ensuremath{\hat{#1}}}

\newcommand{\splittable}{splittable}

\newcommand{\leaves}{{\cal L}}
\newcommand{\compat}{{\cal C}}

\DeclareMathOperator{\load}{load}

\newcommand{\lbelow}[1]{\leaves(#1)}

\newcommand{\ptn}{\mathcal{P}}

\newcommand{\reachu}{\mathcal{B}}

\newcommand{\Lc}{1}
\newcommand{\Rc}{2}

\DeclareMathOperator{\rt}{rt}

\author[1,2]{Neil Olver}
\author[3]{Frans Schalekamp}
\author[5]{Suzanne van der Ster}
\author[2,1,6]{Leen Stougie}
\author[4]{Anke van Zuylen}
\affil[1]{Department of Econometrics \& Operations Research, Vrije Universiteit Amsterdam}
\affil[2]{Centrum Wiskunde \& Informatica}
\affil[3]{School of Operations Reseach and Information Engineering, Cornell University}
\affil[4]{Department of Mathematics, College of William and Mary}
\affil[5]{Albert Heijn Online}
\affil[6]{INRIA-Erable}

\begin{document}
\maketitle
\setcounter{footnote}{0}

\begin{abstract}
We give a 2-approximation algorithm for the Maximum Agreement Forest 
problem on two rooted binary trees. This NP-hard 
problem has been studied extensively in the past two decades, since it can be used to compute the rooted Subtree Prune-and-Regraft (rSPR) distance between two phylogenetic trees. Our algorithm is combinatorial and its running time is quadratic in the input size. To prove the approximation guarantee, we construct a feasible dual solution for a novel linear programming formulation. In addition, we show this linear program is stronger than previously known formulations, and we give a compact formulation, showing that it can be solved in polynomial time.
\end{abstract}

\keywords{Maximum agreement forest, phylogenetic tree, SPR distance, subtree prune-and-regraft distance, computational biology}

\subjclass{68W25, 90C27, 92D15}

\section{Introduction}

	Evolutionary relationships are often modeled by a rooted tree, where the leaves represent a set of species, and internal nodes are (putative) common ancestors of the leaves below the internal node. Such phylogenetic trees date back to Darwin~\cite{Darwin37}, who used them in his notebook to elucidate his thoughts on evolution. For an introduction to phylogenetic trees we refer to \cite{MathEvPhyl, SempleSteel2003} 

The topology of phylogenetic trees can be based on different sources of data, e.g., morphological data, behavioral data, genetic data, etc., which can lead to different phylogenetic trees on the same set of species. %

Such partly incompatible trees may actually be unavoidable: there exist non-tree-like evolutionary processes that preclude the existence of a phylogenetic tree, so-called reticulation events, such as hybridization, recombination and horizontal gene transfer \cite{HusonRuppScornavacca10, Nakhleh2009ProbSolv}. Irrespective of the cause of the conflict, the natural question arises to quantify the dissimilarity between such trees.  
Especially in the context of reticulation, a particularly meaningful measure of comparing phylogenetic trees is the Subtree Prune-and-Regraft distance for rooted trees (rSPR-distance), which provides a lower bound on a certain type of these non-tree evolutionary events. The problem of finding the exact value of this measure for a set of species motivated the formulation of the Maximum Agreement Forest Problem (MAF) by Hein, Jian, Wang and Zhang~\cite{Hein96}.  

In the definition of MAF by Hein et al. we are given two rooted binary trees, each having its leaves labeled with the same set of labels ${\cal L}$, where in each tree each leaf has one label and each label is assigned to one leaf. 
The problem is to find a minimum set of edges to be
deleted from the two trees, so that the rooted trees in the resulting two forests (where the choice of the root is the natural choice)  can be paired up into pairs of \emph{isomorphic} trees. 
Two rooted trees are isomorphic if their \emph{restrictions} are the same, where the restriction of a tree is obtained by considering the minimal tree spanning the same set of leaves from $\leaves$, and then contracting nodes with only a single child.

Since the introduction by Hein et al.\ in \cite{Hein96}, in which they also proved NP-hardness, MAF has been extensively studied, mostly in its version of two rooted binary input trees. After Allen and Steel~\cite{Allen2001} pointed out that the claim by Hein et al.\ that solving MAF on two rooted directed trees computes the rSPR-distance between the trees is incorrect, Bordewich and Semple~\cite{Bordewich04} presented a subtle redefinition of MAF, whose optimal value does coincide with the rSPR-distance. In this redefinition it is required that the two forests agree on the tree containing the original roots of the input trees. This has now become the standard definition of MAF, for which Bordewich and Semple~\cite{Bordewich04} showed that NP-hardness still holds, and Rodrigues~\cite{Rodrigues03} showed that it is in fact \modified{APX}-hard. 

The problem has attracted a lot of attention, and indeed has become a canonical problem in the field of phylogenetic networks. 
Many variants of MAF have been studied, including versions where the input consists of more than two trees~\cite{chataigner2005approximating,chen2016approximating}, and where the input trees are unrooted~\cite{Whidden09,Whidden13} or non-binary~\cite{Rodrigues07, vanIersel14}. 
We will concentrate on MAF in its classical form with two rooted binary input trees, and we will be concerned with the worst-case approximability of the problem.
The literature includes many other approaches to the problem, including fixed-parameter tractable algorithms (e.g.,~\cite{Whidden09,Whidden13}) and integer linear programming~\cite{Wu09,Wu10}.
But the quest for better approximation algorithms has become central within the MAF literature.

Our result improves over a sequence spanning 10 years of approximation algorithms, starting with the first correct  5-approximation~\cite{Bonet06}. This was followed by several 3-approximations~\cite{Bordewich08,Rodrigues07,Whidden13}, each one improving on the running time, and the last one giving a relatively simple and elegant proof. A 2.5-approximation followed~\cite{Shietal15}. 
In 2016, Chen et al.~\cite{chen2016improved} described a $7/3$-approximation.
Independently in the same year,  a subset of the present authors~\cite{DBLP:conf/icalp/SchalekampZS16} gave a factor 2 approximation algorithm.
Subsequently, Chen et al.~\cite{chen2017cubic} gave a different factor 2 approximation algorithm using very different methods, and with a cubic running time.
The 2-approximation algorithm presented in the current paper may be viewed as the full version of the algorithm in~\cite{DBLP:conf/icalp/SchalekampZS16}. %
However, while the algorithm presented here is similar in spirit, it differs in many details, and the exposition is entirely new.
Although the algorithm and analysis remain quite subtle, this version is significantly shorter and clearer.
Moreover, we show how our algorithm can, with some care, be implemented in quadratic time (\cite{DBLP:conf/icalp/SchalekampZS16} discuss only a polynomial time bound).
This improves over the cubic running time of Chen et al.~\cite{chen2017cubic}. %

Our 2-approximation algorithm differs from previous works in two key aspects. 
First of all, our algorithm takes a global approach; choices may depend on large parts of the instance.
All the previous algorithms that obtained a worse approximation ratio considered only local, constant-sized, substructures. 
Secondly, we introduce a novel integer linear programming formulation for the analysis.
Our approximation guarantee is proved by constructing a feasible solution to the dual of this linear program, rather than arguing locally about the objective of the optimal solution. 

While we provide a new integer linear programming formulation and exploit its linear relaxation in our analysis, we do not need to actually solve the relaxation as part of our algorithm.
In fact, the formulation has an exponential number of variables, and so it is not immediately clear that it can be efficiently optimized.
We show that it can be reformulated as a compact LP, with only a polynomial number of variables and constraints.
We believe that this is interesting for a number of reasons.
It implies that the linear relaxation can be solved efficiently (in polynomial time); this may be of future utility in obtaining better approximation guarantees using LP-rounding techniques, which do require an optimal solution to the relaxation.
Moreover, the compact formulation is amenable to use in commercial integer programming solvers.
There is a previous formulation due to Wu~\cite{Wu09}, but our formulation is significantly stronger: the integrality gap of the relaxation of Wu is at least $3.2$, whereas for ours we show it is at most $2$, and in fact the worst example that we are aware of has integrality gap $1.25$ (see the appendix).
Finally, we remark that our compact formulation can be easily adapted to handle other variants of MAF---for example, settings with more than two trees.

We have implemented and tested our algorithm, as well as the compact formulation~\cite{implementation}.
The implementation has been designed so that it is easy to step through the algorithm and explore its behaviour on a given instance; the reader may find it helpful when examining the technical details of the algorithm.

\paragraph*{Outline.} We define the problem and introduce necessary notation in Section~\ref{sec:prelim}. Section~\ref{sec:redblue} describes the algorithm, and proves that it produces a feasible solution to MAF. In Section~\ref{sec:analysis}, we introduce the linear program, and describe a feasible solution to its dual that can be maintained by the algorithm. We then show the objective value of this dual solution is always at least half the objective value of the MAF solution, which proves the approximation ratio of two. In Section~\ref{sec:compact}, we show a compact formulation of the (exponential sized) linear program used for the analysis. 
In the appendices, we show that our algorithm can be implemented to run in time quadratic in the size of the input, and we give an example that shows that a previously known integer linear program~\cite{Wu09} is not as strong as the formulation introduced here.

\section{Preliminaries}\label{sec:prelim}

The input to the Maximum Agreement Forest problem (MAF) consists of two rooted binary trees $T_1$ and $T_2$, where the leaves in each binary tree are labeled with the same label set $\leaves$. Each leaf has exactly one label, and each label in $\leaves$ is assigned to exactly one leaf in $T_1$, and one leaf in $T_2$. We will use $\leaves$ also to denote the leaves of the trees.

Let $V_1$ and $V_2$ denote the node set of $T_1$ and $T_2$ respectively, and let $V = V_1 \cup V_2$.
We call all nodes in $V \setminus \leaves$ \emph{internal nodes}.
We let $\lbelow{u}$ denote the set of leaves that are descendants of a node $u \in V$. 

We will use the following notational conventions: we use $u$ and $v$ to denote arbitrary nodes (including leaves), if the node we refer to is an internal node in $V_2$, we will use $\nodeinTtwo{u}$ and $\nodeinTtwo{v}$, and we use the letters $x, y$ and $w$ to refer to leaves. 

For $A\subset \leaves$ we use $V_i[A]$ to denote the set of (internal) nodes in $T_i$ that lie on a path between any two leaves in $A$ for $i \in \{1,2\}$, and define $V[A] := V_1[A] \cup V_2[A]$.

\begin{definition}
    We will say that a set $A \subseteq \leaves$\ \emph{covers} a node $u \in V$ if $u \in V[A]$.
    We say that $A, A' \subseteq \leaves$ \emph{overlap} if $V[A] \cap V[A'] \neq \emptyset$; we can also say that $A$ \emph{overlaps} $A'$ in $U$, for $U \subseteq V$, if $V[A] \cap V[A'] \cap U \neq \emptyset$.
    We say a partition $\ptn$ of $\leaves$ \emph{overlaps} in $U \subseteq V$ if there exist $A,A'\in \ptn$, $A\neq A'$ such that $A$ and $A'$ overlap in $U$.
\end{definition}

For $A\subseteq \leaves$, we let $\lca_i(A)$ denote the least common ancestor of $A$ in $T_i$. 
We will sometimes omit braces of explicit sets and write, e.g., $\lca_1(x_1,x_2,x_3)$ instead of $\lca_1(\{x_1,x_2,x_3\})$. 

For nodes $u,v$ in the same tree, we use $u\prec v$ to indicate that $u$ is a descendant of $v$ and $u \preceq v$ if $u$ is equal
to $v$ or a descendant of $v$.
\begin{definition}
A set $L\subseteq \leaves$ is \emph{\compatible{}} if for all $x_1,x_2,x_3\in L$
\[\lca_1(x_1,x_2)\prec \lca_1(x_1,x_2,x_3) \Leftrightarrow \lca_2(x_1,x_2)\prec \lca_2(x_1,x_2,x_3).\]
\end{definition}
We call a triple of leaves \emph{in\compatible{}} if it is not a \compatible{} set.
Note that $L \subseteq \leaves$ is compatible precisely if the subtrees induced by $L$ in $T_1$ and $T_2$ are isomorphic.

A feasible solution to MAF is a partition $\ptn = \{ A_1, A_2, \ldots, A_k\}$ of $\leaves$ %
such that every component $A_i$ is compatible, and $A_i$ does not overlap $A_j$, for each $i \neq j$. The cost of this solution is defined to be $|\ptn|-1$.
This cost corresponds to the number of edges that must be deleted from $T_1$, as well as the same number from $T_2$, so that in both of the resulting forests, each $A_i \in \ptn$ is the leaf set of a single tree. 

\begin{wrapfigure}[6]{r}{0.4\textwidth}
    \vspace{-.3cm}
\begin{tikzpicture}[scale=0.6]
\node (T1) at (0.5,2.5) {$T_1$};
\node (B1) at (0,0) {$x_1$};
\node (B2) at (1,0) {$x_2$};
\node (R1) at (2.5,0) {$x_3$};
\node (R2) at (3.5,0) {$x_4$};
\coordinate (l) at (0.5,1) ; \coordinate (ll) at (0.3,1.2);
\coordinate (r) at (3,1) ; \coordinate (rr) at (3.2,1.2);
\node (rho) at (4.25,2) {$\!\!\rho$};
\coordinate (u) at (1.75,2) ; \coordinate (uu) at (1.6,2.2);
\coordinate (top) at (3,3) ;
\path (B1) edge (l)
	(B2) edge (l)
	(R1) edge (r)
	(R2) edge (r)
	(l) edge (u)
	(r) edge (u)
	(u) edge (top)
	(top) edge (rho);
\end{tikzpicture}
\quad
\begin{tikzpicture}[scale=0.6]
\node (T2) at (0.5,2.5) {$T_2$};
\node (B1) at (0,0) {$x_1$};
\node (B2) at (1,0) {$x_3$};
\node (R1) at (2.5,0) {$x_2$};
\node (R2) at (3.5,0) {$x_4$};
\coordinate (l) at (0.5,1) ; \coordinate (ll) at (0.3,1.2);
\coordinate (r) at (3,1) ; \coordinate (rr) at (3.2,1.2);
\node (rho) at (4.25,2) {$\!\!\rho$};
\coordinate (u) at (1.75,2) ; \coordinate (uu) at (1.6,2.2);
\coordinate (top) at (3,3) ;
\path (B1) edge (l)
	(B2) edge (l)
	(R1) edge (r)
	(R2) edge (r)
	(l) edge (u)
	(r) edge (u)
	(u) edge (top)
	(top) edge (rho);
\end{tikzpicture}
\end{wrapfigure}
\paragraph*{Remark.} 
In order for MAF to correspond to the rSPR distance, it is necessary to add an additional node $\rho$ to $\leaves$, as a unique sibling to the root in both $T_1$ and $T_2$.
This is the distinction between the original definition of MAF by Hein~\cite{Hein96} and the correction by Bordewich and Semple~\cite{Bordewich04}.
We simply assume that this addition is already included in the input instance, after which there is no need to distinguish this additional leaf from the others.

To describe our algorithm, it will be convenient to extend the notion of compatibility.
\begin{definition}
Given $K\subseteq \leaves$, we say a set $L\subseteq \leaves$ is \emph{\comp{$K$}} if $L\cap K$ is compatible. A partition $\ptn = \{ A_1, A_2, \ldots, A_k\}$ of $\leaves$  is \emph{\comp{$K$}} if $A_i$ is \comp{$K$} for all $i=1,2,\ldots,k$.
\end{definition}

\section{The Red-Blue Algorithm}\label{sec:redblue}

The algorithm maintains a partition $\ptn$ of $\leaves$, which at the end of the algorithm will correspond to a feasible solution to MAF.
The algorithm will maintain the invariant that $\ptn$ does not overlap in $V_2$. %
Observe that this is equivalent to defining $\ptn$ to be the leaf sets of the trees in a forest, obtained by deleting edges from $T_2$.
Initially $\ptn = \{\leaves\}$.

The algorithm works towards feasibility by iteratively refining $\ptn$, focusing each iteration on a set of leaves $\lbelow{u}$ for some $u\in V_1$, for which the current partition is infeasible in some (quite narrowly defined) way. At the end of the iteration the solution is feasible if we restrict our attention to $\lbelow{u}$, and even if we consider $\lbelow{u}\cup \{w\}$ for any arbitrary $w\in \leaves\setminus \lbelow{u}$.

We use the following definition to specify which sets $\lbelow{u}$ the algorithm considers.

\begin{definition}\label{def:pcs}
Given an infeasible partition $\ptn$ that does not overlap in $V_2$, we call $u\in V_1$ a \emph{\pcs} if at least one of the following holds:
\begin{enumerate}[(a)]
\item $\ptn$ is not \comp{$\lbelow{u}$};
\item $\ptn$ overlaps in $V_1[\lbelow{u}]$;
\item there exists a component $A$ in $\ptn$ such that $A\setminus \lbelow{u}\neq\emptyset$, and $A\cap \lbelow{u} \cup \{w\}$ is not compatible for all $w\in A\setminus \lbelow{u}$.
\end{enumerate}
\end{definition}
Observe that if $u\in V_1$ is  a \pcs{}, then any ancestor of $u$ is a \pcs.
We will say an internal node $u$ in tree $T_i$ is the ``lowest'' node with property $\Gamma$ if property $\Gamma$ does not hold for any of $u$'s descendants in $T_i$.
The algorithm will identify a lowest node $u\in V_1$ that is a \pcs. 

Given a \pcs\ $u\in T_1$, we partition $\leaves$ into $R, B, W$, where $R=\lbelow{u_r}$ and $B=\lbelow{u_\ell}$ for the two children $u_\ell$ and $u_r$ of $u$. We will refer to this partition as a \emph{coloring} of the leaves; we will refer to the leaves in $R$ as red leaves, the leaves in $B$ as blue leaves and the leaves in $W$ as white leaves. We call a component of $\ptn$ \emph{\tricolored{}} if it has a nonempty intersection with $R, B$ and $W$, and \emph{\bicolored{}} if it has a nonempty intersection with exactly two of the sets $R, B, W$. A component is called \emph{\multicolored{}} if it is either \tricolored{} or \bicolored{}, and \emph{\monochrome} otherwise.

\begin{observation}\label{obs:coloring}
Let $u$ be a lowest \pcs{} for $\ptn$, and consider the coloring $R, B, W$, where $R=\lbelow{u_r}$ and $B=\lbelow{u_\ell}$ for the two children $u_\ell$ and $u_r$ of $u$. Then the set of \multicolored{} components of $\ptn$ consists of either at most two \bicolored{} components or exactly one \tricolored{} component.
\end{observation}
\begin{proof}
    If $u$ is a lowest \pcs, $\ptn$ does not overlap in $V_1[R]$ and $V_1[B]$, and so at most one component of $\ptn$ covers $\lca_1(R)$, and at most one covers $\lca_1(B)$.
    The observation thus follows immediately, since any \bicolored{} component covers either $\lca_1(R)$ or $\lca_1(B)$, and any \tricolored{} component covers both $\lca_1(R)$ and $\lca_1(B)$.
\end{proof} 
We note that the above observation can be refined; it is possible to show that $\ptn$ contains  either one \tricolored{} component or exactly two \bicolored{} components; see Lemma~\ref{lem:coloring} in Section~\ref{subsec:obj}.

In Figure~\ref{fig:example}, we give an example of an input $T_1, T_2$ and a coloring of the leaves, where $R_1$ can be a single leaf in $R$, or it can be a tree with leaves labeled by a compatible subset $R_1\subset R$ with likewise interpretations for the other capital leaves in this and future figures; the caption gives three partitions such that $u$ satisfies exactly one of the three conditions of Definition~\ref{def:pcs}:
If $\ptn=\{\leaves\}$, $u$ satisfies (a). Note that $u$ is indeed a lowest \pcs, since $\{R_1,R_2,W_3\}$ and $\{B_1,B_2,W_3\}$ are compatible sets, so $u_\ell$ and $u_r$ do not satisfy (c) (nor (a) and (b)).
If $\ptn = \{\{B_1\}, \{B_2, W_1\}, \{R_1, R_2, W_2, W_3\}\}$, node $u$ satisfies (b). Again, $u$ is a lowest \pcs (clearly $u_\ell$ and $u_r$ does not satisfy (a) and (b);  they also do not satisfy (c) since $\{B_1,W_1\}$ is compatible, as is $\{R_1,R_2,W_3\}$). 
Finally, if $\ptn=\{\{R_1\}, \{B_1, B_2, W_1, R_2, W_2\}, \{W_3\}\}$, node $u$ satisfies (c). Observe that in this case $u$ is again a lowest \pcs. %

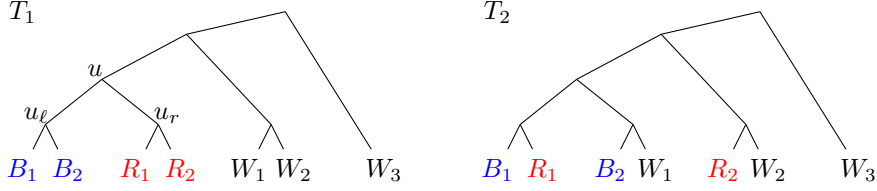
\begin{figure}
\begin{center}
\begin{tikzpicture}[scale=0.6]
\node (T1) at (0,3.5) {$T_1$};
\node[text=blue] (B1) at (0,0) {$B_1$};
\node[text=blue] (B2) at (1,0) {$B_2$};
\node[text=red] (R1) at (2.5,0) {$R_1$};
\node[text=red] (R2) at (3.5,0) {$R_2$};
\node (W1) at (5,0) {$W_1$};
\node (W2) at (6,0) {$W_2$};
\node (W3) at (8,0) {$W_3$};
\coordinate (l) at (0.5,1) ; \node (ll) at (0.3,1.2) {$u_\ell$};
\coordinate (r) at (3,1) ; \node (rr) at (3.2,1.2) {$u_r$};
\coordinate (w) at (5.5,1) ;
\coordinate (u) at (1.75,2) ; \node (uu) at (1.6,2.2) {$u$};
\coordinate (uw) at (3.625,3) ;
\coordinate (root) at (5.8125,3.5) ;
\path (B1) edge (l)
	(B2) edge (l)
	(R1) edge (r)
	(R2) edge (r)
	(l) edge (u)
	(r) edge (u)
	(W1) edge (w)
	(W2) edge (w)
	(w) edge (uw)
	(u) edge (uw)
	(uw) edge (root)
	(W3) edge (root);
\end{tikzpicture}
\qquad
\begin{tikzpicture}[scale=0.6]
\node (T2) at (0,3.5) {$T_2$};
\node[text=blue] (B1) at (0,0) {$B_1$};
\node[text=red] (B2) at (1,0) {$R_1$};
\node[text=blue] (R1) at (2.5,0) {$B_2$};
\node (R2) at (3.5,0) {$W_1$};
\node[text=red] (W1) at (5,0) {$R_2$};
\node (W2) at (6,0) {$W_2$};
\node (W3) at (8,0) {$W_3$};
\coordinate (l) at (0.5,1) ; \coordinate (r) at (3,1) ; 
\coordinate (w) at (5.5,1) ;
\coordinate (u) at (1.75,2) ; \coordinate (uw) at (3.625,3) ;
\coordinate (root) at (5.8125,3.5) ;
\path (B1) edge (l)
	(B2) edge (l)
	(R1) edge (r)
	(R2) edge (r)
	(l) edge (u)
	(r) edge (u)
	(W1) edge (w)
	(W2) edge (w)
	(w) edge (uw)
	(u) edge (uw)
	(uw) edge (root)
	(W3) edge (root);
\end{tikzpicture}
\end{center}
\caption{If $\ptn=\{\leaves\}$, then node $u$ satisfies case (a) of Definition~\ref{def:pcs}; if $\ptn = \{\{B_1\}$, $\{B_2,W_1\}$, $\{R_1,  R_2, W_2, W_3\}\}$, it satisfies case (b) and if $\ptn=\{\{R_1\},$ $\{B_1, B_2, W_1, R_2, W_2\},$ $\{W_3\}\}$, it satisfies (c).}\label{fig:example}
\end{figure}

\begin{center}\fbox{\begin{minipage}[b]{.84\textwidth}
\vspace*{.5\baselineskip}
{\large \quad \sc Red-Blue Algorithm}

\hrulefill
\begin{algorithmic}
\State $\ptn \gets \{\leaves\}$. 
\State $\pairslist \gets \emptyset$.
\While{$\ptn$ is not feasible}
	\State \makebox[0pt][r]{$\star$\quad}Let $u\in T_1$ be a lowest \pcs, with children $u_\ell$ and $u_r$. 
	\State Let $R=\lbelow{u_r}$, $B=\lbelow{u_\ell}$  and $W=\leaves\setminus(R\cup B)$.
	\State {\sc Make-\RBcomp}($\ptn,  (R, B, W)$).%
	\State {\sc Make-Splittable}($\ptn, (R, B, W)$).
	\State {\sc Split}($\ptn, (R, B, W)$).
	\State {\sc Find-Merge-Pair}($\pairslist, \ptn, (R, B, W)$).
\EndWhile
\State {\sc Merge-Components}($\pairslist, \ptn$).

\end{algorithmic}
\end{minipage}}\end{center}
An overview of the algorithm is given above. The procedures will be described in detail in the subsequent subsections, along with with lemmas regarding the properties they ensure. 

We will refer to a pass through the main while-loop of the algorithm as an ``iteration''.
In order to simplify the statement of the lemmas, we will make statements like ``let $\ptn'$ be the partition after {\sc ProcedureName}$(\ptn, (R, B, W))$''. 
 This implicitly assumes that $(R, B, W)$ was a coloring chosen in  the beginning of the current iteration of the Red-Blue algorithm (and thus, that $\lca_1(R\cup B)$ was a lowest \pcs\ at that moment), and that $\ptn'$ is the partition resulting from  calling {\sc ProcedureName}$(\ptn, (R, B, W))$  in the current iteration. 
 
Finally, the $\star$ in front of certain lines will be used to refer to these lines in the analysis in Section~\ref{subsec:dual}.

\subsection{\sc Make-\RBcomp}
\begin{center}\fbox{\begin{minipage}[b]{.84\textwidth}
\begin{algorithmic}
\Procedure{Make-\RBcomp}{$\ptn, (R, B, W)$}

\While{$\exists A\in \ptn$ that is not \RBcomp{}}
	\State \makebox[0pt][r]{$\star$\quad}Let $\nodeinTtwo{u}$ be a lowest internal node in $V_2[A]$ for which $A \cap \lbelow{\nodeinTtwo{u}}$ intersects both $R$ and $B$.
	\State $\ptn \leftarrow \ptn \setminus \{A\} \cup \{A\cap \lbelow{\nodeinTtwo{u}}, A\setminus \lbelow{\nodeinTtwo{u}}\}$.
\EndWhile
\EndProcedure
\end{algorithmic}
\end{minipage}}\end{center}

\begin{figure}
\begin{center}
\begin{tikzpicture}[scale=0.6]
\node (T2) at (0,3.5) {$T_2$};
\node[text=blue] (B1) at (0,0) {$B_1$};
\node[text=red] (B2) at (1,0) {$R_1$};
\node[text=blue] (R1) at (2.5,0) {$B_2$};
\node (R2) at (3.5,0) {$W_1$};
\node[text=red] (W1) at (5,0) {$R_2$};
\node (W2) at (6,0) {$W_2$};
\node (W3) at (8,0) {$W_3$};
\coordinate (l) at (0.5,1) ; \coordinate (r) at (3,1) ; 
\coordinate (w) at (5.5,1) ;
\coordinate (u) at (1.75,2) ; \coordinate (uw) at (3.625,3) ;
\coordinate (root) at (5.8125,3.5) ;
\path (B1) edge (l)
	(B2) edge (l)
	(R1) edge (r)
	(R2) edge (r)
	(l) edge (u)
	(r) edge (u)
	(W1) edge (w)
	(W2) edge (w)
	(w) edge (uw)
	(u) edge (uw)
	(uw) edge (root)
	(W3) edge (root);
\end{tikzpicture}
\newlength{\nl}
\settowidth{\nl}{$\xrightarrow{\text{\sc Make-\RBcomp}}$}
\begin{minipage}[b][2.5cm][c]{\nl}
\begin{center}
$\xrightarrow{\text{\sc Make-\RBcomp}}$
\end{center}
\end{minipage}
\begin{tikzpicture}[scale=0.6]
\node (T2) at (0,3.5) {$T_2$};
\node[text=blue] (B1) at (0,0) {$B_1$};
\node[text=red] (B2) at (1,0) {$R_1$};
\node[text=blue] (R1) at (2.5,0) {$B_2$};
\node (R2) at (3.5,0) {$W_1$};
\node[text=red] (W1) at (5,0) {$R_2$};
\node (W2) at (6,0) {$W_2$};
\node (W3) at (8,0) {$W_3$};
\coordinate (l) at (0.5,1) ; \coordinate (r) at (3,1) ; 
\coordinate (w) at (5.5,1) ;
\coordinate (u) at (1.75,2) ; \coordinate (uw) at (3.625,3) ;
\coordinate (root) at (5.8125,3.5) ;
\node (uprime) at (0.3,1.2) {$\nodeinTtwo{u}$};
\path (B1) edge (l)
	(B2) edge (l)
	(R1) edge (r)
	(R2) edge (r)
	(l) edge[dashed] (u)
	(r) edge (u)
	(W1) edge (w)
	(W2) edge (w)
	(w) edge (uw)
	(u) edge (uw)
	(uw) edge (root)
	(W3) edge (root);
\end{tikzpicture}
\end{center}
\caption{Illustration of {\sc Make-\RBcomp}($\ptn,(R,B,W)$). Because $\ptn$ and $\ptn'$ do not overlap in $V_2$, we can represent these as the leaf sets of trees in a forest obtained by deleting edges from $T_2$. In this figure and the following figures the dashed edges represent deleted edges. 
\\
In this example $\ptn=\{\leaves\}$. Then $\nodeinTtwo{u}=\lca_2(R_1, B_1)$, and {\sc Make-\RBcomp}($\ptn,(R,B,W)$) refines the partition to $\{\{B_1, R_1\}, \{B_2, W_1, R_2, W_2,W_3\}\}$, which is \RBcomp{}.}\label{fig:examplecompatible}
\end{figure}
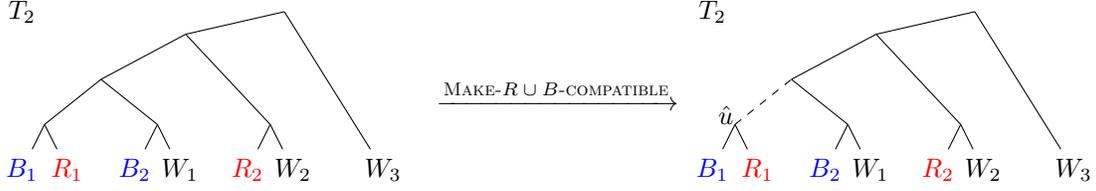
An example is given in Figure~\ref{fig:examplecompatible}. We note that in general, the choice of $\nodeinTtwo{u}$ does not have to be unique, and that multiple refinements may be needed  to make the partition \RBcomp{}. 

As observed above, for any partition $\ptn$ that does not overlap in $V_2$, there is a set of edges in $T_2$, such that $\ptn$ consists of the leaf sets of the trees in the forest obtained after deleting these edges. Our refinement is equivalent to deleting the edge from $\nodeinTtwo{u}$ towards the root in $T_2$ and hence the resulting partition does not overlap in $V_2$ if the original partition did not overlap in $T_2$.

\begin{lemma}\label{lem:RBcomp}
Let $\ptn'$ be the partition after {\sc Make-\RBcomp}$(\ptn, (R, B,W))$. Then $\ptn'$ is a refinement of $\ptn$ that does not overlap in $V_2$ and is \RBcomp{}.
\end{lemma}

\begin{proof}
First, observe $\ptn$ is \comp{$R$} and \comp{$B$}, since $u$'s children are not \pcsplural.  
If $\ptn$ is \RBcomp{} then $\ptn$ is not modified by the procedure, and the lemma is vacuously true. Otherwise, the procedure refines $\ptn$, and we already mentioned above that the resulting partition $\ptn'$ does not overlap in $V_2$ provided that $\ptn$ does not overlap in $V_2$.
The procedure ends when there are no sets in $\ptn$ that are not \RBcomp{}, so the only thing left to show is that this procedure halts. Because $\nodeinTtwo{u}$ was chosen to be the lowest internal node in $V_2[A]$ such that $A \cap \lbelow{\nodeinTtwo{u}}$ intersects both $R$ and $B$, the children of $\nodeinTtwo{u}$, say $\nodeinTtwo{u}_r$ and $\nodeinTtwo{u}_\ell$, are so that $A \cap \lbelow{\nodeinTtwo{u}_r}$ and $A \cap \lbelow{\nodeinTtwo{u}_\ell}$ can only intersect one of $R$ and $B$. Therefore  $A\cap \lbelow{\nodeinTtwo{u}}$ is \RBcomp{}, where $A$ was not, and thus the number of \RBcomp{} components in $\ptn$ increases, which can only happen  at most  $|\leaves|$ times. 
\end{proof}

Observe that if $\ptn$ is \RBcomp{}, then any refinement of $\ptn$ is also \RBcomp{}, hence we may assume that the partition at any later point in the current iteration of the Red-Blue Algorithm is \RBcomp{}.

\subsection{\sc Make-Splittable}
The goal of the next two procedures is to further refine the partition so that there is no overlap in $V_1[R\cup B]$. We will do this in two steps, the first of which will achieve the following property.

\begin{definition}
Given a coloring $(R,B,W)$ of $\leaves$. A set $A \subseteq \leaves$ is {\em \splittable} if $A \cap R$, $A\cap B$ and $A\cap W$ do not overlap in $V_2$.
\end{definition}

\begin{center}\fbox{\begin{minipage}[b]{.84\textwidth}
\begin{algorithmic}
\Procedure{Make-Splittable}{$\ptn, (R, B, W)$}

\While{$\exists A\in \ptn$ that is not \splittable{}}
	\State \makebox[0pt][r]{$\star$\quad}Let $\nodeinTtwo{u}$ be a lowest internal node in $V_2[A]$ such that $A \cap \lbelow{\nodeinTtwo{u}}$ is \bicolored{} and $A\setminus \lbelow{\nodeinTtwo{u}}$ intersects precisely the same colors as $A$.
	\State $\ptn \gets \ptn \setminus \{A\} \cup \{A\cap \lbelow{\nodeinTtwo{u}}, A\setminus \lbelow{\nodeinTtwo{u}}\}$.
\EndWhile
\EndProcedure
\end{algorithmic}
\end{minipage}}\end{center}

First, note that by the same arguments as in the previous subsection, the partition that results from {\sc Make-Splittable} does not overlap in $V_2$ if the original partition did not overlap in $V_2$.
 It is easy to see that if $A$ is \bicolored{} and not \splittable{}, then there exists $\nodeinTtwo{u}\in V_2[A]$ such that both $A\cap \lbelow{\nodeinTtwo{u}}$ and $A\setminus \lbelow{\nodeinTtwo{u}}$ are \bicolored{}:  just take $\nodeinTtwo{u}$ to be a lowest node in $V_2[A\cap C_1] \cap V_2[A\cap C_2]$ for distinct $C_1,C_2\in \{R, B, W\}$. We prove below in Lemma~\ref{lem:splittablewelldefined} that if $A$ is \tricolored{}, we can additionally ensure that $A\setminus\lbelow{\nodeinTtwo{u}}$ is \tricolored{}. For this to hold, we need that $\ptn$ is \RBcomp{}, which by Lemma~\ref{lem:RBcomp} is indeed true when {\sc Make-Splittable} is called. 

\begin{figure}
\begin{center}
\begin{tikzpicture}[scale=0.6]
\node (T2) at (0,3.5) {$T_2$};
\node[text=blue] (B1) at (0,0) {$B_1$};
\node[text=red] (B2) at (1,0) {$R_1$};
\node[text=blue] (R1) at (2.5,0) {$B_2$};
\node (R2) at (3.5,0) {$W_1$};
\node[text=red] (W1) at (5,0) {$R_2$};
\node (W2) at (6,0) {$W_2$};
\node (W3) at (8,0) {$W_3$};
\coordinate (l) at (0.5,1) ; \coordinate (r) at (3,1) ; 
\coordinate (w) at (5.5,1) ;
\coordinate (u) at (1.75,2) ; \coordinate (uw) at (3.625,3) ;
\coordinate (root) at (5.8125,3.5) ;
\path (B1) edge (l)
	(B2) edge[dashed] (l)
	(R1) edge (r)
	(R2) edge (r)
	(l) edge (u)
	(r) edge (u)
	(W1) edge (w)
	(W2) edge (w)
	(w) edge (uw)
	(u) edge (uw)
	(uw) edge (root)
	(W3) edge[dashed] (root);
\end{tikzpicture}
\settowidth{\nl}{$\xrightarrow{\text{\sc Make-\RBcomp}}$}
\begin{minipage}[b][2.5cm][c]{\nl}
\begin{center}
$\xrightarrow{\text{\sc Make-Splittable}}$
\end{center}
\end{minipage}
\begin{tikzpicture}[scale=0.6]
\node (T2) at (0,3.5) {$T_2$};
\node[text=blue] (B1) at (0,0) {$B_1$};
\node[text=red] (B2) at (1,0) {$R_1$};
\node[text=blue] (R1) at (2.5,0) {$B_2$};
\node (R2) at (3.5,0) {$W_1$};
\node[text=red] (W1) at (5,0) {$R_2$};
\node (W2) at (6,0) {$W_2$};
\node (W3) at (8,0) {$W_3$};
\coordinate (l) at (0.5,1) ; \coordinate (r) at (3,1) ; 
\coordinate (w) at (5.5,1) ;
\coordinate (u) at (1.75,2) ; \coordinate (uw) at (3.625,3) ;
\coordinate (root) at (5.8125,3.5) ;
\node (uprime) at (3.3,1.2) {$\nodeinTtwo{u}$};
\path (B1) edge (l)
	(B2) edge[dashed] (l)
	(R1) edge (r)
	(R2) edge (r)
	(l) edge (u)
	(r) edge[dashed] (u)
	(W1) edge (w)
	(W2) edge (w)
	(w) edge (uw)
	(u) edge (uw)
	(uw) edge (root)
	(W3) edge[dashed] (root);
\end{tikzpicture}
\end{center}
\caption{Illustration of {\sc Make-Splittable}($\ptn,(R,B,W)$). $\ptn=\{\{R_1\}, \{B_1, B_2, W_1, R_2, W_2\}, \{W_3\}\}$, and the set $A=\{B_1, B_2, W_1, R_2, W_2\}$ is not \splittable{}. {\sc Make-Splittable}($\ptn$) would choose $\nodeinTtwo{u}=\lca_2(B_2, W_1)$ and replace $A$ by $\{B_2,W_1\}$ and $\{B_1, R_2, W_2, W_3\}$.
}\label{fig:examplesplittable}
\end{figure}
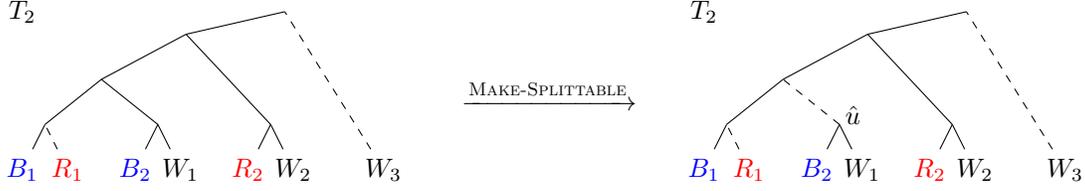
As a first example of {\sc Make-Splittable}, consider $\ptn=\{\{B_1, R_1\}, \{B_2, W_1, R_2, W_2,W_3\}\}$ that was the output of {\sc Make-\RBcomp} depicted in Figure~\ref{fig:examplecompatible}. In this example $\ptn$ is already \splittable{}. In Figure~\ref{fig:examplesplittable} a more interesting example is given.

\begin{lemma}\label{lem:splittablewelldefined}
    \textsc{Make-\splittable} is well-defined, in that a node $\nodeinTtwo{u}$ satisfying the desired properties in line $\star$ can always be found. %
\end{lemma}
\begin{proof}
As noted above the existence of $\nodeinTtwo{u}$ is clear when $A$ is \bicolored{}. 
So suppose $A$ is \tricolored{} and not \splittable{}. Note that $V_2[A\cap R]$ and $V_2[A\cap B]$ cannot intersect because $A$ is \RBcomp{}. Assume without loss of generality that $V_2[A\cap R]\cap V_2[A\cap W]\neq \emptyset$, and let $\nodeinTtwo{u}$ be a lowest node in $V_2[A\cap R]\cap V_2[A\cap W]$.  Note that both $A\cap \lbelow{\nodeinTtwo{u}}$ and $A\setminus \lbelow{\nodeinTtwo{u}}$ must intersect $W$ and $R$, and that $A\cap \lbelow{\nodeinTtwo{u}}$ cannot intersect $B$, since then $A$ is not \RBcomp{}. So $A\cap \lbelow{\nodeinTtwo{u}}$ is \bicolored{}%
, and $A\setminus \lbelow{\nodeinTtwo{u}}$ is \tricolored{}.
\end{proof}

\begin{lemma}\label{lem:splittable}
Let $\ptn'$ be the partition after {\sc Make-\splittable}$(\ptn, (R, B,W))$. Then $\ptn'$ is a refinement of $\ptn$ that does not overlap in $V_2$ and  in which every component is \splittable{}.
\end{lemma}
\begin{proof}
    By Lemma~\ref{lem:splittablewelldefined}, and since each iteration increases the number of components in $\ptn$, {\sc Make-\splittable} must terminate, and by its definition, the final partition $\ptn'$ contains only splittable components.
    Clearly $\ptn'$ is a refinement of $\ptn$; it does not overlap in $V_2$ by the same arguments as used in the proof of Lemma~\ref{lem:RBcomp}.
\end{proof}

Before continuing, we summarize the properties of the partition that is the result after {\sc Make-Splittable} that will be useful in the proof of the approximation guarantee in Section~\ref{sec:analysis}. 
To describe these, we need the notion of a \emph{top component}. 
\begin{definition}
\label{def:topcomponent}
Given the partition $\ptn$ at the start of the current iteration, and $\ptn'$ another partition encountered in the current iteration, let ${\cal D}$ be the components that were created during the current iteration, i.e., ${\cal D}=\ptn'\setminus \ptn$. Then $A\in {\cal D}$ is a \emph{top component} if there exists no $A'\in {\cal D}$ such that $\lca_2(A)\prec \lca_2(A')$. 
\end{definition}

\begin{lemma}\label{lem:properties}
Let $\ptnstart$ denote the partition at the start of a given iteration, and $(R,B,W)$ the coloring of the leaves that is selected, let $\ptnRBcomp$  denote the partition after {\sc Make-\RBcomp($\ptnstart, (R, B, W)$)}, and let $\ptnsplit$ denote the partition after
{\sc Make-Splittable($\ptnRBcomp,(R, B, W)$)},
\begin{enumerate}
\item\label{itm:multi0} Only \multicolored{} components are subdivided by the iteration, i.e., if $A\in\ptnstart\setminus \ptnsplit$, then $A$ is \multicolored.
\item\label{itm:multi1} Only \multicolored{} components are created by {\sc Make-\RBcomp} and {\sc Make-Splittable}, i.e., if $A\in \ptnsplit\setminus \ptnstart$, then $A$ is \multicolored.
\item\label{itm:trinum} 
The number of \tricolored{} components in $\ptnsplit$ is the same as in $\ptnRBcomp$.
\item\label{itm:trinotop}
    Any \tricolored{} component in $\ptnRBcomp$ or $\ptnsplit$ that is not a top component contains no \compatible{} \tricolored{} triple.
\item\label{itm:binotop}
    Any \bicolored{} component $A$ in $\ptnsplit$ that is not a top component satisfies that $\lca_2(A)$ is not overlapped by $A\cap C$ for any color $C\in \{R, B, W\}$. In other words, $\lbelow{\nodeinTtwo{u}_\ell} \cap A$ and $\lbelow{\nodeinTtwo{u}_r} \cap A$ are \monochrome{} where $\nodeinTtwo{u}_\ell$ and $\nodeinTtwo{u}_r$ are the children of $\lca_2(A)$.  
\item\label{itm:wtop}
    If $x_W$ is in a top component $A$ in $\ptnstart$ and $x_W$ is not a descendant of $\lca_2(A\cap (R\cup B))$, then $x_W$ is in a top component in $\ptnsplit$. \end{enumerate}
\end{lemma}
\begin{proof}
Each of the properties is easily verified by inspection of the {\sc Make-\RBcomp} and {\sc Make-Splittable} procedures.
\modified{For example, point \ref{itm:trinotop} follows from the fact that a node $\nodeinTtwo{u}$ picked in {\sc Make-\RBcomp} is always chosen as low as possible.
This implies that for the newly created component $A'$, and any $r \in A' \cap R, b \in A' \cap B$, $\lca_2(r,b) = \nodeinTtwo{u}$.}
\end{proof}

\subsection{\sc Split}
The next procedure  will refine $\ptn$ so that the resulting partition does not overlap in $V_1[R\cup B]$. Since by Lemma~\ref{lem:splittable}, $\ptn$ is \splittable{}, we can simply intersect each component with $R$, $B$ and $W$, to achieve this property. However, we will need to be slightly more careful in order to achieve the approximation guarantee; in particular, we will sometimes need to perform what we call a {\sc Special-Split}.

\begin{center}\fbox{\begin{minipage}[b]{.84\textwidth}
\begin{algorithmic}
\Procedure{Split}{$\ptn, (R, B, W)$}
\For{each \multicolored{} component $A$}
	\If{$A$ is \tricolored{}, and there exists a \tricolored{} triple in $A$ that is \compatible{}}
		\State {\sc Special-Split($A, \ptn, (R,B,W)$)}
	\Else{} 
		\State $\ptn \gets \ptn \setminus \{A\} \cup \{A\cap R, A\cap B, A\cap W\}$ (where empty sets are not added)
 	\EndIf
\EndFor					 
\EndProcedure
\end{algorithmic}
\end{minipage}}
\end{center}
 \begin{remark}
Our analysis in Section~\ref{sec:analysis} needs the {\sc Special-Split}, {\sc Find-Merge-Pair} and {\sc Merge-Components} procedures only in one (of three) cases that will be described in Lemma~\ref{lem:coloring}. Without these procedures, it is trivial to see  that the resulting partition is feasible, and we will see in Section~\ref{sec:analysis} that the proof of the approximation ratio is quite simple in these cases. 
\end{remark}

We now describe the property that the outcome partition of {\sc Split} will have, which goes beyond merely being \RBcomp{} and non-overlapping in $V_2\cup V_1[R\cup B]$. We first define that property, and give  necessary and sufficient conditions for a partition that does not overlap in $V_2$ to have this property. 
\begin{definition}
    Let $K \subseteq \leaves$. A partition $\ptn$ is \emph{\feas{$K$}} if for all $w\in \leaves$, $\ptn$ is \comp{$K\cup \{w\}$}, and no two components in $\ptn$ overlap in $V_2\cup V_1[K]$.
\end{definition}
We will simply say $\ptn$ is {\em feasible} if it is \feas{$\leaves$}, which we note does indeed coincide with the definition of a feasible solution to MAF. We make two additional remarks about the notion of $K$-feasibility:
\begin{itemize}
\item
Being \feas{$K$} requires something stronger than simply not overlapping in $V_2\cup V_1[K]$ and $K$-compatibility. The stronger compatibility notion will be used in Lemma~\ref{lem:merge1} to show that if $\ptn$ is \feas{$R\cup B$}, then future iterations of the Red-Blue algorithm will not further subdivide (the partition induced on) the leaves in $R\cup B$. This is not necessarily true if $\ptn$ is only \RBcomp{} and does not overlap in $V_2\cup V_1[R\cup B]$.\footnote{{For example, if $\ptn$ has only one \multicolored{} component, which is \RBcomp{} but not \splittable{}, and $\ptn$ does not overlap in $V_2\cup V_1[R\cup B]$, then  the Red-Blue algorithm will further subdivide the partition induced on $R\cup B$.} }

\item
If $u\in V_1$ is a \pcs{} for $\ptn$, then $\ptn$ is not \feas{$\lbelow{u}$}. The converse is not true, however: if $\ptn$ contains a single component containing $\lbelow{u}$ which is \comp{$\lbelow{u}$}, but this component contains both $w\in \leaves\setminus\lbelow{u}$ such that $\lbelow{u}\cup \{w\}$ is compatible, and $w'\in \leaves\setminus\lbelow{u}$ such that $\lbelow{u}\cup \{w'\}$ is not compatible, then $\ptn$ is {\em not} \feas{$\lbelow{u}$}, but $u$ is not a \pcs.
The stronger notion of a $u$ being a \pcs{} versus not being \feas{$\lbelow{u}$} is needed when we prove the approximation guarantee in Section~\ref{sec:analysis}.
\end{itemize}

The following technical lemma gives conditions to check if a partition does not overlap in $V_1[R\cup B]$.

\begin{lemma}\label{lem:ugh}
Let $\ptn$  be the partition and $(R, B, W)$ be the coloring at the start of an iteration.
    Let $\ptn'$ be a refinement of $\ptn$ that does not overlap in $V_2$ and that is \RBcomp{}. Then
    $\ptn'$ does not overlap in $V_1[R\cup B]$ if     
    \begin{itemize}
    \item[(i)]$\ptn'$ has at most one \multicolored{} component $A^*$;
    \item[(ii)]if $A^*$ exists, and $\lca_2(A^*) \prec \lca_2(R\cup B)$, then any node $v'$ with $\lca_2(A^*) \prec v' \preceq \lca_2(R\cup B)$ is covered only by components in $\ptn'$ that are subsets of $W$, or that are also components of $\ptn$.
    \end{itemize}
\end{lemma}

\begin{proof}
    Suppose the conditions of the lemma hold for $\ptn'$.
First, observe that by (i), $\ptn'$ contains at most one component covering $\lca_1(R\cup B)$. 
Suppose for a contradiction that $A',A''\in \ptn'$ overlap in $V_1[R]\cup V_1[B]$. 

Since $\lca_1(R\cup B)$ was chosen as a lowest \pcs, $\lca_1(R)$ and $\lca_1(B)$ were not \pcsplural\ for $\ptn$. This implies that  no two components of $\ptn$ overlap in $V_1[R]\cup V_1[B]$, so it must be the case that $A'$ and $A''$ were both part of a single component in $\ptn$. 
Furthermore, $\ptn$ must have been \comp{$R$} and \comp{$B$}, so $(A'\cup A'')\cap R$ and $(A'\cup A'') \cap B$ are compatible sets. We will show that these facts imply that if $A'$ and $A''$ overlap in $V_1[R]$ or $V_1[B]$, then they must overlap in $V_2[R]$ or $V_2[B]$ respectively, thus contradicting that $\ptn'$ does not overlap in $V_2$.

Let $v$ be a lowest node in $V_1[R\cup B]$ such that $A'\cap \lbelow{v} \neq \emptyset$ and $A''\cap \lbelow{v}\neq\emptyset$ (where we note that $v$ exists since $A',A''$ overlap in some node in $V_1[R\cup B]$). 
Observe that a child of $v$ cannot be in both $V_1[A']$ and $V_1[A'']$, as this contradicts the choice of $v$. Hence $v$ can be in $V_1[A']$ and $V_1[A'']$ only if $A'$ and $A''$ also contain leaves in $\leaves\setminus \lbelow{v}$.
Let $x',x''$ be in $A'\cap \lbelow{v}$ and $A''\cap\lbelow{v}$ respectively, and choose $y',y''$ in $A' \setminus \lbelow{v}$ and $A''\setminus \lbelow{v}$. 

First, assume both $A'$ and $A''$ are \monochrome{}, and thus they are each either red or blue (since otherwise they would not overlap in $V_1[R]\cup V_1[B]$).
Note that $\lca_1(x',x'') = v \prec \lca_1(x',x'',y')$ and similarly $\lca_1(x',x'') \prec \lca_1(x',x'',y'')$. 
Since $\{x',x'',y',y''\}$ is a compatible set, we must also have $\lca_2(x',x'') \prec \lca_2(x',x'',y')$ and $\lca_2(x',x'') \prec \lca_2(x',x'',y'')$. But then $\lca_2(x',x'')$ is on the path from $x'$ to $y'$ as well as on the path from $x''$ to $y''$. Hence, $A'$ and $A''$ overlap in $\lca_2(x',x'')\in V_2$, contradicting that   $\ptn'$ does not overlap in $V_2$.

Now, suppose $A'$ is \monochrome, and $A''$ is the unique \multicolored{} component $A^*$, and $A',A^*$ overlap in $V_1[R]\cup V_1[B]$. Without loss of generality $A'\subseteq R$. Then we still know that $\{x',x'',y'\}$ is compatible, and thus that $\lca_2(x',x'') \prec \lca_2(x',x'',y')$, so that $\lca_2(x',x'')\in V_2[A']$.
Now, $x''\in A^*$ is a descendant of $\lca_2(x',x'')$, so if $A^*$ also has a leaf that is not a descendant of $\lca_2(x',x'')$ then $A'$ and $A^*$ overlap in $\lca_2(x',x'')$, again contradicting that $\ptn'$ does not overlap in $V_2$.
So it must be the case that $\lca_2(A^*)$ is a descendant of $\lca_2(x',x'')$. But then $V_2[A']$ intersects the path from $\lca_2(A^*)$ to $\lca_2(R\cup B)$. But we already showed above that $A'\cup A^*$ was part of a single component in $\ptn$, which implies that $A'$ is not a component of $\ptn$, contradicting (ii).
\end{proof}

We now describe the {\sc Special-Split} procedure. 
Recall that this is only called if $A$ is \tricolored{}, and there is at least one \tricolored{} compatible triple in $A$.

\begin{center}\fbox{\begin{minipage}[b]{.84\textwidth}
\begin{algorithmic}
\Procedure{Special-Split}{$A,\ptn, (R, B, W)$}
\If{every tricolored triple in $A$ is \compatible{}}
	\State $\ptn \gets \ptn \setminus \{A\} \cup \{A\cap R, A\setminus R\}$.	
\Else{}
	\State \makebox[0pt][r]{$\star$\quad}Let $\nodeinTtwo{u} = \lca_2(A\cap (R\cup B))$.
	\State $\ptn \gets \ptn \setminus \{A\} \cup \{A\setminus \lbelow{\nodeinTtwo{u}},A'\cap R, A'\cap B, A'\cap W\}$ where $A'= A\cap \lbelow{\nodeinTtwo{u}}$.	
\EndIf
\EndProcedure
\end{algorithmic}
\end{minipage}}\end{center}

\begin{figure}
\begin{center}
\begin{tikzpicture}[scale=0.6]
\node (T2) at (0,3.5) {$T_2$};
\node[text=blue] (B1) at (0,0) {$B_1$};
\node[text=red] (B2) at (1,0) {$R_1$};
\node[text=blue] (R1) at (2.5,0) {$B_2$};
\node (R2) at (3.5,0) {$W_1$};
\node[text=red] (W1) at (5,0) {$R_2$};
\node (W2) at (6,0) {$W_2$};
\node (W3) at (8,0) {$W_3$};
\coordinate (l) at (0.5,1); 
\coordinate (r) at (3,1); 
\coordinate (w) at (5.5,1);
\coordinate (u) at (1.75,2); 
\coordinate (uw) at (3.625,3);
\coordinate (root) at (5.8125,3.5);
\path (B1) edge (l)
	(B2) edge[dashed] (l)
	(R1) edge (r)
	(R2) edge (r)
	(l) edge (u)
	(r) edge[dashed] (u)
	(W1) edge (w)
	(W2) edge (w)
	(w) edge (uw)
	(u) edge (uw)
	(uw) edge (root)
	(W3) edge[dashed] (root);
\end{tikzpicture}
\settowidth{\nl}{$\xrightarrow{\text{\sc Make-\RBcomp}}$}
\begin{minipage}[b][2.5cm][c]{\nl}
\begin{center}
$\xrightarrow{\text{\qquad \sc Split \qquad}}$
\end{center}
\end{minipage}
\begin{tikzpicture}[scale=0.6]
\node (T2) at (0,3.5) {$T_2$};
\node[text=blue] (B1) at (0,0) {$B_1$};
\node[text=red] (B2) at (1,0) {$R_1$};
\node[text=blue] (R1) at (2.5,0) {$B_2$};
\node (R2) at (3.5,0) {$W_1$};
\node[text=red] (W1) at (5,0) {$R_2$};
\node (W2) at (6,0) {$W_2$};
\node (W3) at (8,0) {$W_3$};
\coordinate (l) at (0.5,1);
\coordinate (r) at (3,1) ; 
\coordinate (w) at (5.5,1) ;
\coordinate (u) at (1.75,2) ; \coordinate (uw) at (3.625,3) ;
\coordinate (root) at (5.8125,3.5) ;
\path (B1) edge[dashed] (l)
	(B2) edge[dashed] (l)
	(R1) edge[dashed] (r)
	(R2) edge (r)
	(l) edge (u)
	(r) edge[dashed] (u)
	(W1) edge[dashed] (w)
	(W2) edge (w)
	(w) edge (uw)
	(u) edge (uw)
	(uw) edge (root)
	(W3) edge[dashed] (root);
\end{tikzpicture}
\end{center}
\vspace*{\baselineskip}

\begin{center}
\begin{tikzpicture}[scale=0.6]
\node (T2) at (0,3.5) {$T_2$};
\node[text=blue] (B1) at (0,0) {$B_1$};
\node[text=red] (B2) at (1,0) {$R_1$};
\node[text=blue] (R1) at (2.5,0) {$B_2$};
\node (R2) at (3.5,0) {$W_1$};
\node[text=red] (W1) at (5,0) {$R_2$};
\node (W2) at (6,0) {$W_2$};
\node (W3) at (8,0) {$W_3$};
\coordinate (l) at (0.5,1) ; \coordinate (r) at (3,1) ; 
\coordinate (w) at (5.5,1) ;
\coordinate (u) at (1.75,2) ; \coordinate (uw) at (3.625,3) ;
\coordinate (root) at (5.8125,3.5) ;
\path (B1) edge (l)
	(B2) edge (l)
	(R1) edge (r)
	(R2) edge (r)
	(l) edge[dashed] (u)
	(r) edge (u)
	(W1) edge (w)
	(W2) edge (w)
	(w) edge (uw)
	(u) edge (uw)
	(uw) edge (root)
	(W3) edge (root);
\end{tikzpicture}
\settowidth{\nl}{$\xrightarrow{\text{\sc Make-\RBcomp}}$}
\begin{minipage}[b][2.5cm][c]{\nl}
\begin{center}
$\xrightarrow{\text{\qquad \sc Split \qquad}}$
\end{center}
\end{minipage}
\begin{tikzpicture}[scale=0.6]
\node (T2) at (0,3.5) {$T_2$};
\node[text=blue] (B1) at (0,0) {$B_1$};
\node[text=red] (B2) at (1,0) {$R_1$};
\node[text=blue] (R1) at (2.5,0) {$B_2$};
\node (R2) at (3.5,0) {$W_1$};
\node[text=red] (W1) at (5,0) {$R_2$};
\node (W2) at (6,0) {$W_2$};
\node (W3) at (8,0) {$W_3$};
\coordinate (l) at (0.5,1) ; \coordinate (r) at (3,1) ; 
\coordinate (w) at (5.5,1) ;
\coordinate (u) at (1.75,2) ; \coordinate (uw) at (3.625,3) ;
\node (uprime) at (3.625,3.4) {$\nodeinTtwo{u}$};
\coordinate (root) at (5.8125,3.5) ;
\path (B1) edge[dashed] (l)
	(B2) edge (l)
	(R1) edge[dashed] (r)
	(R2) edge (r)
	(l) edge[dashed] (u)
	(r) edge (u)
	(W1) edge[dashed] (w)
	(W2) edge (w)
	(w) edge (uw)
	(u) edge (uw)
	(uw) edge[dashed] (root)
	(W3) edge (root);
\end{tikzpicture}
\end{center}
\caption{Two illustrations of {\sc Split}($\ptn,(R,B,W)$). 
In the top example $\ptn=\{\{R_1\},$ $\{B_2,W_1\},$ $\{B_1,R_2,W_2\},$ $\{W_3\}\}$ and {\sc Split}($\ptn$) would simply refine each set of $\ptn$ by intersecting it with the three color classes. The result is that every leaf is a singleton in $\ptn'$.
\\
In the bottom example, $\ptn=\{\{B_1,R_1\}, \{B_2, W_1, R_2, W_2, W_3\}\}$. The set $A=\{B_2, W_1, R_2, W_2, W_3\}$ is \tricolored{} and contains triple $\{B_2, R_2, W_3\}$ that is \tricolored{} and \compatible{}, but not every \tricolored{} triple in $A$ is \compatible{}, e.g. $\{B_2,R_2, W_2\}$ is not \compatible{}.  In this case, the {\sc Special-Split} replaces $A$ by $\{\{B_2\}, \{R_2\},$ $\{W_1,W_2\},$ $\{W_3\}\}$.
}\label{fig:examplesplit}
\end{figure}
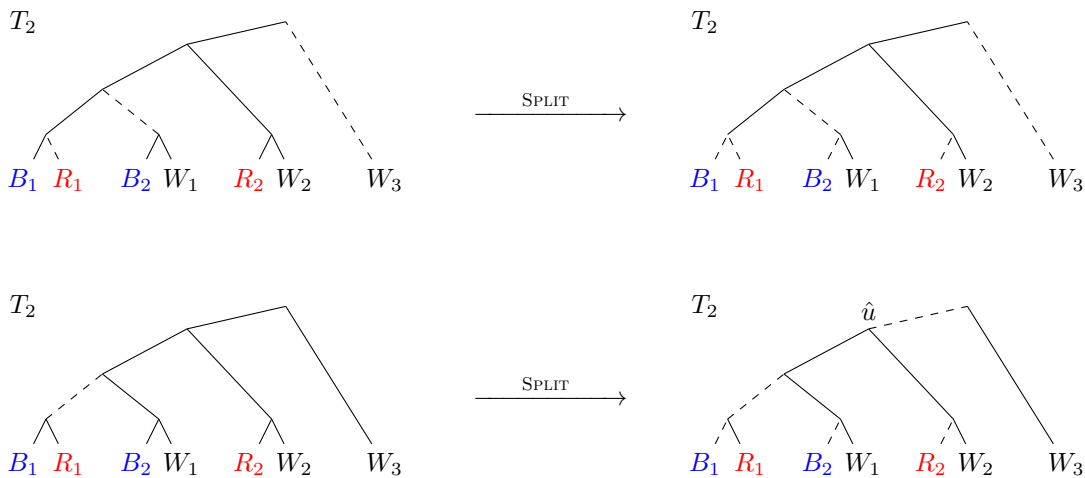

The next lemma states that this ensures the partition resulting after {\sc Split} is \feas{$R\cup B$}.
\begin{lemma}\label{lem:feas}
Let $\ptn'$ be the partition after {\sc Split}$(\ptn, (R, B,W))$. Then $\ptn'$ is a refinement of $\ptn$ that is \feas{$R\cup B$}.
Moreover, {\sc Special-Split} is applied to at most one component in each iteration of the Red-Blue algorithm, implying that $\ptn'$ has at most one \multicolored{} component.
\end{lemma}

\begin{proof}
The fact that $\ptn'$ does not overlap in $V_2$ follows from the fact that every component of $\ptn$ was \splittable{}. 

It is also easy to see that every component is \comp{$R \cup B\cup \{w\}$} for all $w\in \leaves$: each component is either \monochrome{} (and thus \comp{$R\cup B\cup \{w\}$} by the fact that the partition is \RBcomp\ by Lemma~\ref{lem:RBcomp}), or it is the result of a {\sc Special-Split} on a component that was already \comp{$R \cup B\cup \{w\}$} for all  $w\in \leaves$ 
before the {\sc Special-Split}.

It remains to show that no two components in $\ptn'$ overlap in $V_1[R\cup B]$.
By Lemma~\ref{lem:ugh}, it suffices to show $\ptn'$ has at most one \multicolored{} component, and that this component, if it exists, is a top component (recall Definition~\ref{def:topcomponent}).
Note that the only possible \multicolored{} components of $\ptn'$ are \bicolored{} components  created by {\sc Special-Split} on a component $A\in \ptn$ that is \tricolored{} and in which every \tricolored{} triple is \compatible{}. By property~\ref{itm:trinotop} of Lemma~\ref{lem:properties}, the only \tricolored{} components that have a  compatible \tricolored{} triple are top components, and by Observation~\ref{obs:coloring}, the partition at the start of the iteration had at most one \tricolored{} component, and thus there is also at most one \tricolored{} top component in $\ptn$. So $\ptn'$ has at most one \multicolored{} component, which is a top component, and by Lemma~\ref{lem:ugh}, this implies $\ptn'$ does not overlap in $V_1[R\cup B]$.
\end{proof}

\subsection{{\sc Find-Merge-Pair} and {\sc Merge-Components}}\label{subsec:merge}

The astute reader may have noted that the Red-Blue Algorithm sometimes increases the number of components by more than necessary to be \feas{$R\cup B$}.
For example, it follows from the arguments in the proof of the previous lemma that if there is a \tricolored{} component in which every \tricolored{} triple is \compatible{}, then not further subdividing this component would also leave a partition that is \feas{$R\cup B$}. 
{\sc Find-Merge-Pair} and {\sc Merge-Components} aim to merge two components of the partition produced at the end of {\sc Split}, so that
the partition with the merged components is still \feas{$R\cup B$}.
{\sc Find-Merge-Pair} thus looks for a pair of components that can be merged, by scanning the components of the current partition, and finding two leaves in $R\cup B$ that are in different sets of the partition now, but that were in the same component at the start of the current iteration. 

\begin{center}\fbox{\begin{minipage}[b]{.84\textwidth}
\begin{algorithmic}
 \Procedure{Find-Merge-Pair}{$\pairslist, \ptn, (R, B, W)$}
   \If{exists $x_1,x_2\in R\cup B$ such that
	   \begin{itemize}[\hspace{.9cm}]\setlength{\itemsep} {0pt}\setlength{\partopsep}{0pt}\setlength{\topsep}{0pt}%
  		\item $x_1$ and $x_2$ were in the same component at the start of the current iteration, 
		\item $x_1$ and $x_2$  are in distinct components $A_1$ and $A_2$ in $\ptn$, and 
        \item $\ptn\setminus \{A_1,A_2\}\cup \{A_1\cup A_2\}$ is \feas{$R\cup B$} 	
        \end{itemize}
	}
  	\State $\pairslist \gets \pairslist \cup \{(x_1,x_2)\}$
  \EndIf
 \EndProcedure
\end{algorithmic}
\end{minipage}}\end{center}

Although we could simply merge the components containing $x_1$ and $x_2$ for the pair found by {\sc Find-Merge-Pair}, we will not do so until the very end of the algorithm.
The reason we keep such ``superfluous'' splits is because they will increase the objective value of the dual solution we use to prove the approximation guarantee of 2 (see Section~\ref{sec:analysis}). We ``reverse'' these superfluous splits (i.e., we will merge components) at the end of the algorithm; this is reminiscent of a ``reverse delete'' in approximation algorithms for network design~\cite{GoemansW97}. %

\begin{center}\fbox{\begin{minipage}[b]{.84\textwidth}
\begin{algorithmic}
 \Procedure{Merge-Components}{$\pairslist, \ptn$}
\For{each pair $(x_1,x_2)$ in \pairslist}

	\State 	Let $A_1$ and $A_2$ be the sets in $\ptn$ containing $x_1$ and $x_2$, respectively.
	\State $\ptn \gets \ptn \setminus\{A_1, A_2\}\cup \{A_1\cup A_2\}$.
\EndFor
 \EndProcedure
\end{algorithmic}
\end{minipage}}\end{center}

The proof that we will be able to merge the components containing  the pair of leaves identified by {\sc Find-Merge-Pair} at the end of the algorithm will rely on the fact that (i) because the partition is \comp{$R\cup B\cup \{w\}$} for any $w\in \leaves$, merging the components containing the identified leaves $x_1,x_2\in R\cup B$ cannot increase the number of incompatible triples contained in a component, and (ii) because the partition is \feas{$R\cup B$}, future iterations of the algorithm will not further refine the partition induced on $R\cup B$. 
This  is the reason why we do not allow {\sc Find-Merge-Pair} to choose leaves in $W$ (and only choosing leaves in $R\cup B$ is sufficient to prove the claimed approximation guarantee).

\begin{lemma}\label{lem:merge1}
Let $(R, B, W)$ be the coloring during some iteration of the Red-Blue algorithm, let $\ptn$ be the partition at the end of the pass, and let $x,x'\in R\cup B$. If $x,x'$ are in the same component of $\ptn$, then $x$ and $x'$ are in the same component in any partition at any later point of the algorithm's execution.  
\end{lemma}
\begin{proof}
	Let $(R',B',W')$ be the coloring of the leaves in some later iteration of the algorithm, and suppose for a contradiction that the iteration with coloring $(R',B',W')$ separates $x$ and $x'$  in different components.
    From a brief consideration of the algorithm, it is apparent that there must exist some $\nodeinTtwo{u}\in T_2$ such that $A\cap \lbelow{\nodeinTtwo{u}}$ is \multicolored{} with respect to the coloring $(R',B',W')$, and $A \cap \lbelow{\nodeinTtwo{u}}$ contains precisely one of $x,x'$. 
    By relabeling if needed, assume that $x \in A \cap \lbelow{\nodeinTtwo{u}}$ and $x' \in A \setminus \lbelow{\nodeinTtwo{u}}$, and let $w \in A \cap \lbelow{\nodeinTtwo{u}}$ be any leaf with a color different from $x$, and note that 
    \begin{equation}
    \lca_2(x,w)\prec \nodeinTtwo{u} \prec\lca_2(x,x',w).\label{eq:contra}
    \end{equation}

    Since $\ptn$ is \feas{$R\cup B$}, no $v\in V_1[R\cup B]$ is a \pcs, and hence all leaves in $R\cup B$, and in particular $x$ and $x'$, must have the same color in the coloring $(R',B',W')$.     Furthermore, if $w$ has a different color than $x$ and $x'$ in $(R',B',W')$, then $w\not \in R\cup B$, and thus $\lca_2(x,x')\prec \lca_2(x,x',w)$. 
    But, since $\ptn$ is \comp{$R\cup B\cup \{w\}$}, this implies that if $w$ is in the same component as $x$ and $x'$ in (a refinement of) $\ptn$, then $\lca_2(x,x')\prec \lca_2(x,x',w)$, contradicting (\ref{eq:contra}), because only one of $\lca_2(x,x')$ and $\lca_2(x,w)$ can be strictly below $\lca_2(x,x',w)$. 
    \end{proof}

\subsection{Correctness of the algorithm}\label{subsec:MAFfeasible}
\begin{theorem}
The Red-Blue Algorithm returns a feasible solution to MAF.
\end{theorem}
\begin{proof}
Let $k$ be the number of pairs in \pairslist. We prove the theorem by induction on $k$. If $k=0$, then the algorithm returns the partition obtained at the end of the while-loop, which is feasible by the fact that otherwise $\lca_1(\leaves)$ would be a \pcs.

If $k >0$, observe that the final partition is the same irrespective of the order in which the pairs in \pairslist{} are considered. We  may thus assume without loss of generality that they are considered in the reverse order in which they were added to \pairslist. Let $\ptn'$ be the partition after the components have been merged for all pairs on \pairslist, except the pair $(x_1,x_2)$ that was added to \pairslist{} first. Let $\ptn$ be the partition at the moment when $(x_1,x_2)$ was added to \pairslist, and let $R, B, W$ be the three color sets at that moment. Observe that $\ptn'$ is a refinement of $\ptn$,
and that, by Lemma~\ref{lem:merge1}, $\ptn$ and $\ptn'$ induce the same partition of $R\cup B$.

Let $A_1, A_2$ be the components in $\ptn$ containing $x_1,x_2$ respectively. By the choice of $x_1,x_2$, $(A_1\cup A_2)$ is $R\cup B\cup \{w\}$-compatible for any $w\in \leaves$, and does not overlap any component of $\ptn\setminus\{A_1,A_2\}$.

If $A_1, A_2$  are \monochrome{}, they both contain leaves in $R\cup B$ only, and thus by Lemma~\ref{lem:merge1}, $\ptn'$ contains components $A_1$ and $A_2$ as well. Furthermore, in this case, the set $A_1\cup A_2$ is a subset of $R\cup B$ and thus the fact that it is $R\cup B\cup \{w\}$-compatible for any $w\in \leaves$ implies it is \compatible{}. The fact that $A_1\cup A_2$ does not overlap any set $A\in \ptn\setminus\{A_1,A_2\}$ implies it also does not overlap any set $A'\in \ptn'\setminus\{A_1,A_2\}$, since $\ptn'$ is a refinement of $\ptn$.

If $A_1$ and $A_2$ are not both \monochrome{}, observe that only one of $A_1, A_2$ is \bicolored{} and contains leaves in $B\cup W$, since those are the only type of \multicolored{} components after {\sc Split}, and $\ptn$ does not overlap in $V_1[R\cup B]$ so it can only have one \multicolored{} component.
Suppose without loss of generality that $A_1$ is \monochrome{} and $A_2$ contains leaves in $B\cup W$. 
By Lemma~\ref{lem:merge1}, $\ptn'$ contains component $A_1$ and a component $A_2'\subseteq A_2$, where $A_2'\cap (R\cup B) = A_2\cap (R\cup B)$.

We need to show that $A_1\cup A_2'$ is \compatible{} and does not overlap any component in $\ptn'\setminus\{A_1,A_2'\}$.
For the latter, suppose in order to derive a contradiction that $A_1\cup A_2'$ overlaps $A'\in \ptn'\setminus\{A_1,A_2'\}$. 
Observe that the only nodes in $V[A_1\cup A_2']$ that are not in $V[A_1]\cup V[A_2']$ are in $V_2\cup V_1[R\cup B]$, so the overlap must be on a node $v\in V_2\cup V_1[R\cup B]$.
Since $\ptn'$ is a refinement of $\ptn$, there must exist $A\in \ptn$ such that $A'\subset A$, and thus $A_1\cup A_2'$ overlaps $A$ in $v$ as well. But then $A_1\cup A_2$ also overlaps $A$ in $v$ contradicting that $\ptn\setminus\{A_1,A_2\}\cup \{A_1\cup A_2\}$  is \feas{$R\cup B$}.

To show that $A_1\cup A_2'$ is \compatible{}, note that $A_2'$ is \compatible{}, and that $A_1\cup A_2'\subset A_1\cup A_2$ is \comp{$R\cup B\cup \{w\}$} for any $w\in \leaves$. 
So to show that $A_1\cup A_2'$ is \compatible{}, it suffices to consider $x,w,w'\in A_1\cup A_2'$ with $x\in A_1$ and $w,w'\in A_2'\cap W$.

Fix any $\blueleaf{} \in A_2'\cap B$.
Note that $\lca_i( \blueleaf{}, w ) = \lca_i( x,\blueleaf{},w ) =\lca_i( x, w )$ for $i=1,2$, because $\lca_i( x, \blueleaf{} ) \prec \lca_i( x, \blueleaf{}, w )$ because $A_1\cup A_2$ is \comp{$R\cup B\cup\{w\}$}.
Therefore $\{x,w,w'\}$ is \compatible{} exactly when $\{\blueleaf{},w,w'\}$ is \compatible{}. We conclude that $A_1\cup A_2'$ is \compatible{} because $A_2'$ is \compatible{}.
\end{proof}

\section{Proof of the approximation guarantee}\label{sec:analysis}
We showed in the previous section that the Red-Blue algorithm returns a feasible solution $\ptn$.
In order to prove that our algorithm achieves an approximation guarantee of 2, we will use linear programming duality. 

\subsection{The linear programming relaxation}
Introduce a variable $x_L$ for every compatible set $L\in \compat$, where in an integral solution, $x_L = 1$ indicates that the tree with leaf set $L$ forms part of the solution to MAF.
The constraints ensure that in an integral solution, $\{ L : x_L = 1\}$ is a partition, and that $V[L] \cap V[L'] = \emptyset$ for two distinct sets $L, L'$ with $x_L = x_{L'} = 1$.
The objective encodes the size of the partition minus 1.
\begin{equation}\tag{LP}\label{eq:lp} 
\begin{array}{lll}
\minimize &\sum_{L\in \compat} x_L - 1,\\
\mbox{s.t.}& \sum_{L:v\in L} x_L = 1 & \forall v \in \leaves,\\
&\sum_{L:v\in V[L]} x_L \le 1 & \forall v \in V\setminus \leaves,\\
&x_L \geq 0 &\forall L\in \compat.
\end{array}
\end{equation}

In fact, it will be convenient for our analysis to expand the first set of constraints to contain a constraint for every (not necessarily compatible) set of leaves $A$, stating that every such set must be intersected by at least one tree in the chosen MAF solution. 
All these constraints are clearly already implied by the constraints for $A$ a singleton, already present in \eqref{eq:lp}, but they provide us a more expressive dual.
\begin{equation}\tag{LP$'$}\label{eq:lpbig}
    \begin{array}{lll}
\minimize &\sum_{L\in \compat} x_L - 1,\\
\mbox{s.t.}& \sum_{L:A\cap L\neq \emptyset} x_L \ge 1 & \forall A \subseteq \leaves, A \neq \emptyset\\
&\sum_{L:v\in V[L]} x_L \le 1 & \forall v \in V\setminus \leaves,\\
&x_L \ge 0 &\forall L\in \compat.
\end{array}
\end{equation}
The dual of \eqref{eq:lpbig} is 
\begin{equation}\tag{D$'$}\label{LP:dual} 
    \begin{array}{lll}
\max &\sum_{v \in V \setminus \leaves} y_v + \sum_{A\subseteq \leaves}  z_A - 1,\\
\mbox{s.t. }&\sum_{v\in V[L] \setminus \leaves} y_v +\sum_{A:A\cap L\neq \emptyset}z_A \le 1 & \forall L\in \compat,\\
&y_v \le 0& \forall v\in V\setminus \leaves,\\
&z_A \ge 0 &\forall A\subseteq\leaves.
\end{array}
\end{equation}
We will refer to the left-hand side of the first family of constraints, i.e., $\sum_{v\in V[L] \setminus \leaves} y_v +\sum_{A:A\cap L\neq \emptyset}z_A$, as the \emph{load} on set $L$, and denote it by $\load_{(y,z)}(L)$.
By weak duality, we have that the objective value of any feasible dual solution provides a lower bound on the objective value of any feasible solution to \eqref{eq:lp}, and hence also on the optimal value of any feasible solution to MAF. 
Hence, in order to prove that an agreement forest that has $|\ptn|$ components is a 2-approximation, it suffices to find a feasible dual solution with objective value $\frac12(|\ptn|-1)$, i.e., for every new component created by the algorithm, the dual objective value should increase by $\frac12$ (on average).

\subsection{The dual solution}\label{subsec:dual}
The dual solution maintained is as follows. 
Throughout the main loop of the algorithm, $z_A=1$ if and only if $A$ is a component in $\ptn$. In the last part of the algorithm, when we merge components according to \pairslist{}, we do not update the dual solution; these operations affect the primal solution (i.e., $\ptn$) only.

Initially, $y_v = 0 $ for all $v\in V_1\cup V_2$. 
At the start of each iteration, we decrease $y_{u}$ by $1$, where $u=\lca_1(R\cup B)$. %
Whenever in the algorithm we choose a component $A$ and a node $\nodeinTtwo{u}\in V_2[A]$, and separate the component $A$  into $A\cap \lbelow{\nodeinTtwo{u}}$ and $A\setminus \lbelow{\nodeinTtwo{u}}$, we decrease $y_{\nodeinTtwo{u}}$ by $1$. To be precise this happens in {\sc Make-\RBcomp, Make-\splittable} and in one case in {\sc Special-Split} (where we actually further refine $A\cup\lbelow{\nodeinTtwo{u}}$). 
The lines where such nodes are chosen are indicated by $\star$ in the description of the algorithm  and the procedures it contains. 

\begin{lemma}
The dual solution maintained by the algorithm is feasible.
\end{lemma}
\begin{proof}
We prove the lemma by induction on the number of iterations.
Initially, $z_A=0$ for all $A\neq \leaves$ and $z_\leaves=1$ and hence every compatible set $L$ has a load of 1.

At the start of an iteration, we decrease $y_{\lca_1(R\cup B)}$ by $1$, thus decreasing the load by  1 on any \multicolored{} compatible set $L$. 
We show that the remainder of the iteration increases the load by at most 1 on a \multicolored{} compatible set and that it does not increase the load on any \monochrome{} compatible set.

First, observe that  {\sc Make-\RBcomp} and {\sc Make-Splittable} do not increase the load on any set: Separating $A$ into $A\cap \lbelow{\nodeinTtwo{u}}$ and $A\setminus \lbelow{\nodeinTtwo{u}}$ increases the load on sets $L$ that intersect both $A\cap \lbelow{\nodeinTtwo{u}}$ and $A\setminus \lbelow{\nodeinTtwo{u}}$, since $z_A$ gets decreased from 1 to 0, and $z_{A\cap\lbelow{\nodeinTtwo{u}}}$ and $z_{A\setminus\lbelow{\nodeinTtwo{u}}}$ increase from 0 to 1. However, in this case $\nodeinTtwo{u}\in V[L]$, and thus decreasing $y_{\nodeinTtwo{u}}$ by 1 ensures that the load on $L$ does not increase.

To analyze the effect of {\sc Split}, we use the following two claims.
\begin{claim}
In the procedure {\sc Split}$(\ptn, (R, B, W))$
the load on any compatible set $L$ is  increased by at most the number of components $A\in \ptn$ such that $L\cap A$ is \multicolored{}.
\end{claim}
\begin{proof_claim}
If the load on $L$ is increased because
{\sc Split} splits a \bicolored{} component $A$ into two \monochrome{} components, then $L$ must intersect both new components, so $L\cap A$ is \bicolored{} (and thus \multicolored{}).

Consider the case where the load on $L$ is increased because a \tricolored{} component $A$ is split into $A\cap R$, $A\cap B$ and $A\cap W$. This split happens when {\em all} \tricolored{} triples in $A$ are in\compatible{}. Therefore $L\cap A$ cannot be \tricolored{}, and the load is increased by 1. And again $L\cap A$ is \multicolored{}.

Suppose the load on $L$ is increased because {\sc Special-Split}($A, \ptn, (R, B, W)$) is executed for a component~$A$. We consider the two cases. Either $A$ is split into two components, one of which contains all red leaves in $A$. The load on a set $L$  thus increases by 1 if $L\cap A$ is \multicolored{} and $L\cap A\cap R\neq\emptyset$ and by 0 otherwise.
If $A$ is split into four components; we think of this as first splitting $A$ into $A\cap \lbelow{\nodeinTtwo{u}}$ and $A\setminus \lbelow{\nodeinTtwo{u}}$, and then splitting $A\cap \lbelow{\nodeinTtwo{u}}$ by intersecting with $R, B$ and $W$. Since $y_{\nodeinTtwo{u}}$ is decreased by 1, splitting $A$ into $A\cap \lbelow{\nodeinTtwo{u}}$ and $A\setminus \lbelow{\nodeinTtwo{u}}$ does not affect the load on any set $L$.
Splitting $A\cap \lbelow{\nodeinTtwo{u}}$ by intersecting with $R, B, W$ increases the load on $L$ by 1 if $L\cap A\cap \lbelow{\nodeinTtwo{u}}$ is \bicolored{} and by 2 if it is \tricolored{}; note however that the latter is impossible,  since $\nodeinTtwo{u}=\lca_2(A\cap (R\cup B))$,  so any \tricolored{} triple in $A\cap \lbelow{\nodeinTtwo{u}}$  must be in\compatible{}. So the load on $L$ again increases by at most 1 if $A\cap L$ is \multicolored{}.
\end{proof_claim}
\begin{claim}
If $L$ is compatible, and $A$ and $A'$ do not overlap in $V_2$, then $L\cap A$ and $L\cap A'$ cannot both be \multicolored{}.
\end{claim}
\begin{proof_claim}
Since $V_2[A]$ and $V_2[A']$ are disjoint, it must be the case that 
$\lca_2(x,y)\prec \lca_2(x,y,x')$ for all $x,y\in A$ and $x'\in A'$, or $\lca_2(x',y')\prec \lca_2(x,x',y')$  for all $x',y'\in A'$ and $x\in A$ (or both).
Hence, if $L\cap A$ and $L\cap A'$ are both \multicolored{} sets, then there exists $x,y,x',y'\in L$ where $x,y$ have different colors, $x',y'$ have different colors, $\lca_2(x,y)\prec \lca_2(x,y,x')$, and $\lca_2(x,y)\prec \lca_2(x,y,y')$.
We claim this implies $\{x,y,x',y'\}$ is incompatible.

Clearly one of $x,y$ has the same color as one of $x',y'$.
Suppose first that either red or blue is a shared color.
Without loss of generality, we may assume that $x$ and $x'$ are both red; $y$ is then either blue or white. $x$ and $x'$ being red implies $\lca_1(x,x')\prec \lca_1(x,y,x')$, which, since $\lca_2(x,y)\prec \lca_2(x,y,x')$, shows that $\{x,x',y\}$ is an in\compatible{} triple.

So suppose that white is the only shared color, and that $x$ and $x'$ are both white.
Then either $y$ is red and $y'$ is blue, or vice versa. This implies  $\lca_1(y,y')\prec \lca_1(x,y,y')$, and so, since $\lca_2(x,y)\prec \lca_2(x,y,y')$, this implies  $\{x,y,y'\}$ is an in\compatible{} triple.
\end{proof_claim}

It follows immediately from the two claims that {\sc Split} increases the load by at most 1 on any \multicolored{} compatible set and that it does not increase the load on any \monochrome{} set, which completes the proof of the lemma.
\end{proof}

\subsection{The primal and dual objective values}\label{subsec:obj}

Let $\ptn$, \pairslist{} be the partition and \pairslist{} at the end of an iteration, and let $D=\sum_{v\in V\setminus \leaves} y_v + |\ptn|-1$ be the objective value of the dual solution at this time. In this section, we show that our algorithm maintains the invariant that
\begin{equation}
2D\ge \left(|\ptn|-1-|\pairslist|\right).\label{eq:balance}
\end{equation}
Observe that the approximation guarantee immediately follows from this inequality, since the objective value of the algorithm's solution is $\ptn-1-|\pairslist|$ (where $\ptn$, \pairslist{} are the partition and \pairslist{} at the end of the final iteration), and by weak duality $D$ gives a lower bound on the optimal value of the MAF instance.

To prove that the algorithm maintains the invariant, we will show that a given iteration increases the left-hand side of (\ref{eq:balance}) by at least as much as the right-hand side. 
We let $\Delta D$ be the change in the dual objective during the iteration and $\Delta P$ be the increase in the number of components less the number of pairs added to \pairslist{} (either 0 or 1) during the current iteration.

Since at the start of the algorithm, the partition consists of exactly one component, and $y_v=0$ for all $v\in V\setminus \leaves$, (\ref{eq:balance}) holds before the first iteration. So to show (\ref{eq:balance}), it suffices to show that 
\begin{equation}
2\Delta D \ge \Delta P\label{eq:balance2}
\end{equation}
for any iteration.

\bigskip

In what follows, we use the following to refer to the state of the partition at various points in the current iteration: 
$\ptnstart$ at the start; $\ptnRBcomp$ after {\sc Make-\RBcomp}; $\ptnsplit$ after {\sc Make-Splittable}; and $\ptnfeas$ after {\sc Split}.

We begin by showing that the coloring $(R,B,W)$ and the partition $\ptnstart$  satisfies the conditions of one of three cases. 
\begin{lemma}\label{lem:coloring}
Given an infeasible partition $\ptnstart$ that does not overlap in $T_2$, let $u\in V_1$ be a  lowest \pcs, and let $u_\ell$ and $u_r$ be $u$'s children in $T_1$.
Let $R=\lbelow{u_r}, B=\lbelow{u_\ell}$, and $W=\leaves \setminus (R\cup B)$. Then $\ptnstart$  is \comp{$R$} and \comp{$B$} and satisfies exactly one of the following three additional properties:
\begin{description}
\item[{\bf Case 1.}] 
$\ptnstart$ has exactly one \multicolored{} component, say $A_0$, where $A_0$ is \tricolored{}, not \RBcomp{}, and there exists $\whiteleaf\in A_0\setminus \lbelow{\lca_2(A_0\cap (R\cup B))}$.
\item[{\bf Case 2.}] $\ptnstart$ has exactly two \multicolored{} components, say $A_B, A_R$, where $A_B\cap R=\emptyset$ and $A_R\cap B=\emptyset$.
\item[{\bf Case 3.}] $\ptnstart$ has exactly one \multicolored{} component, say $A_0$, where $A_0$ is \tricolored{}, \RBcomp{} and $A_0$ contains no \compatible{} \tricolored{} triple. 
\end{description}
\end{lemma}
We will see in the proof below that Case 1, 2 and 3 correspond to a lowest \pcs{} satisfying (a), (b) and (c) respectively in Definition~\ref{def:pcs}. We refer the reader to Figure~\ref{fig:example} for an illustration of the three cases.
\begin{proof}
    Observe that if $\ptnstart$ is infeasible, then the root of $T_1$, i.e., $\lca_1(\leaves)$ is a \pcs, and that any $v\in \leaves$ is not a \pcs. Hence, $u$ is well-defined and $R$ and $B$ are non-empty.
Note that $\ptnstart$ is \comp{$R$} and \comp{$B$}, since $u$'s children are not \pcs. 

We will show that if $u$ satisfies condition (a) in the definition of \pcs, then the conditions of  Case 1 are satisfied, if (b) holds, the conditions of Case 2 are satisfied, and if (c) holds, but not (a), then the conditions of Case 3 are satisfied. %

We start with (b). Observe that, because $\ptnstart$ is \comp{$R$} and \comp{$B$}, there must be at least two \multicolored{} components if (b) holds. 
If there are two \multicolored{} components, both containing, say, red leaves, then they overlap in $u_r=\lca_1(R)$, which implies $u_r$  is a \pcs, contradicting the choice of $u$. Similarly, there is at most one \multicolored{} component containing blue leaves. Hence, the conditions of Case 2 are satisfied.

If (b) does not hold, then there is at most one \multicolored{} component; the conditions in (a) and (c) both imply there must be at least one (and thus there is exactly one), which we will call $A_0$. We let $R_0=R\cap A, B_0=B\cap A$ and $\nodeinTtwo{u}=\lca_2(R_0\cup B_0)$ (where we stress that $\nodeinTtwo{u}$ is a node in $V_2$, whereas $u$ is a node in $V_1$).

If (a) holds, then $A_0$ is not \RBcomp{}, and thus $R_0\neq \emptyset, B_0\neq\emptyset$. %
Assume, in order to derive a contradiction, that $A_0\subseteq \lbelow{\nodeinTtwo{u}}$.
Observe that, because $A_0$ is not \RBcomp{}, $\lca_2(R_0)=\nodeinTtwo{u}$ or $\lca_2(B_0) = \nodeinTtwo{u}$. Suppose the former without loss of generality. But then $\lca_1(R_0)$ is a \pcs{} satisfying (c) which is a descendant of $u$ thus contradicting the choice of $u$: $A_0\setminus \lbelow{\lca_1(R_0)}= A_0\setminus R_0\supseteq B_0 \neq\emptyset$, and $R_0\cup \{w\}$ is incompatible for any $w\in A_0\setminus R_0$. 
 
Suppose now (c) holds, but not (a), i.e., $\ptnstart$ is \RBcomp{}. Thus $A_0$ must be \RBcomp{} and $A_0\setminus (R_0\cup B_0)\neq\emptyset$. Assume without loss of generality that $R_0\neq\emptyset$, and note that $\lca_1(R_0)$ is a descendant of $u$, and that, if $B_0=\emptyset$, then (c) holds for $\lca_1(R_0)$, contradicting the choice of $u$. Hence, $A_0$ is \tricolored{}. Since  $A_0$ is \RBcomp{}, $\lca_2(R_0)$ and $\lca_2(B_0)$ must be descendants of the distinct children of $\nodeinTtwo{u}$, or the children itself. Furthermore, the fact that $\ptnstart$ is not \comp{$R\cup B \cup \{w\}$} for \emph{any} $w\in A_0\setminus (R\cup B)$ implies that all white leaves are descendants of $\nodeinTtwo{u}$ as well, and thus any \tricolored{} triple of leaves in $A_0$ is in\compatible{}. Thus, if (c) holds but not (a), the conditions for Case 3 are satisfied.
\end{proof}

Recall that the coloring is defined only at the start of the iteration. The lemma ensures that the partitions during the iteration always have either one (in Case 1 and 3) or two (in Case 2) top components.

For Cases 2 and 3, the analysis is quite simple.

\begin{proposition}\label{prop:23}
Let the initial partition $\ptnstart$ and coloring $(R, B, W)$ satisfy the conditions of Case 2 or 3 in Lemma~\ref{lem:coloring}.
Then
$2\Delta D \ge \Delta P$.
\end{proposition}
\begin{proof}
We first make two observations that apply in Case 2 and 3: (i) $\ptnstart$ is already \RBcomp{}, so $\ptnRBcomp=\ptnstart$, and (ii) {\sc Split}($\ptnsplit,(R,B,W)$) will not perform any {\sc Special-Split}, by property~\ref{itm:trinotop} in Lemma~\ref{lem:properties} and  because the top component has no \tricolored{} triple that is \compatible{} (since we are in Case 2 or 3).

These two observations imply that 
\begin{equation}|\ptnfeas| - |\ptnsplit| = |\ptnsplit|-|\ptnstart|+2.\label{eq:primal}\end{equation}
To see this, note that, since no {\sc Special-Split} is performed, $|\ptnfeas|-|\ptnsplit|$ is equal to the number of \bicolored{} components in $\ptnsplit$ plus twice the number of \tricolored{} components in $\ptnsplit$.
Since $\ptnRBcomp=\ptnstart$, and using  properties~\ref{itm:multi1} and~\ref{itm:trinum} of Lemma~\ref{lem:properties}, $\ptnsplit$ has $|\ptnsplit|-|\ptnstart|$ more \multicolored{} components than $\ptnstart$, and the same number of \tricolored{} components as $\ptnstart$. So in Case 2, $\ptnsplit$ has $|\ptnsplit|-|\ptnstart| +2$ \bicolored{} components and zero \tricolored{} components, and in Case 3, $\ptnsplit$ has $|\ptnsplit|-|\ptnstart|$ \bicolored{} components plus one \tricolored{} component, thus indeed (\ref{eq:primal}) holds. 

In addition,  we note that
\begin{equation}
\Delta D =|\ptnfeas|-|\ptnsplit|-1.\label{eq:dual}
\end{equation}
To see this, note that at the start of the iteration, the dual objective value is reduced by $1$ when $y_{u}$ is decreased by $1$ for $u=\lca_1(R\cup B)$.
{\sc Make-\splittable} does not change the dual objective value, because, even though $|\ptn|$ increases by 1 every time the number of components increases by 1, $\sum_v y_v$ decreases by 1 as well. 
Finally, since {\sc Split} will not perform any {\sc Special-Split}, the increase in the dual objective value due to {\sc Split} is equal to the increase in the number of components due to {\sc Split}, which is $|\ptnfeas|-|\ptnsplit|$.

Note that the size of \pairslist{} may increase but will never decrease, 
and thus 
\begin{align*}
\Delta P &\le |\ptnfeas|-|\ptnsplit|+|\ptnsplit|-|\ptnstart|&\\
&= 2\left(|\ptnfeas|-|\ptnsplit|\right) - 2 &&\mbox{by (\ref{eq:primal})}\\
&= 2\Delta D &&\mbox{by (\ref{eq:dual})}.
\end{align*}
\end{proof}

We now prove a similar proposition for Case 1, the proof of which is more involved.

\begin{proposition}\label{prop:1}
Suppose the initial partition $\ptnstart$ and coloring $(R, B, W)$ satisfy the conditions of Case 1 in Lemma~\ref{lem:coloring}.
Then
$2\Delta D \ge \Delta P$.
\end{proposition}
\begin{proof}
In Case 1, we start with one \tricolored{} component $A_0$, which is not \RBcomp{}. Let $\whiteleaf$ be a white leaf in $A_0$ that is not a descendant of $\lca_2(A_0\cap (R\cup B))$, which exists by the definition of  Case 1.
By property~\ref{itm:wtop} in Lemma~\ref{lem:properties}, $\whiteleaf$ is also contained in the top component of $\ptnfeas$, and by property~\ref{itm:multi1} in Lemma~\ref{lem:properties}, the top component is \multicolored{}. Therefore, the top component of $\ptnsplit$ is either \bicolored{}, or it is \tricolored{} and a  {\sc Special-Split} is performed on the top component.

Let $\chi$ be an indicator variable that is 1 if the top component in $\ptnsplit$
is \tricolored{} and has a \tricolored{} triple that is in\compatible{}. 
Let $t$ be the number of \tricolored{} components in $\ptnsplit$ that are not top components.
We claim that
\begin{equation}
|\ptnfeas|-|\ptnsplit| = |\ptnsplit|-|\ptnstart|+1+2\chi + t.
\label{eq:primal1}
\end{equation}
Indeed, if $\chi =0$, the top component is divided into two components by {\sc Split}, and if $\chi=1$ it is subdivided into four components. Thus splitting the top component increases the number of components by $1+2\chi$. 
By property~\ref{itm:multi1} of Lemma~\ref{lem:properties}, $\ptnsplit$ has $|\ptnsplit|-|\ptnstart|$ \multicolored{} components that are not top components, and by property~\ref{itm:trinotop}, each of the \tricolored{} components that are not top components do not require a {\sc Special-Split} and are thus subdivided into three components by {\sc Split}. Hence, splitting the components that are not top components increases the number of components by $|\ptnsplit|-|\ptnstart|+t$.

Next, we analyze the increase in the dual objective. We claim that
\begin{equation}
\Delta D = |\ptnfeas|-|\ptnsplit| - 1 -\chi.
\label{eq:dual1}
\end{equation}To see this, note that the dual objective is decreased by 1 when we decrease $y_{\lca_1(R\cup B)}$ by $1$ at the start of the iteration. The dual objective is not affected by {\sc Make-\RBcomp} and {\sc Make-Splittable}. Finally, if $\chi = 0$, the increase in the dual objective due to {\sc Split} is equal to the increase in the number of components $|\ptnfeas|-|\ptnsplit|$. If $\chi=1$, the same holds, but {\sc Special-Split} on the top component also decreases $y_{\nodeinTtwo{u}_0}$ by 1.

So we get that \begin{align*}
|\ptnfeas|-|\ptnstart| &= |\ptnfeas|-|\ptnsplit| + |\ptnsplit|-|\ptnstart|\\
&= 2\left(|\ptnfeas|-|\ptnsplit|\right) -1-2\chi- t  &&\mbox{ by (\ref{eq:primal1})}\\
&= 2\Delta D + 1 - t  &&\mbox{ by (\ref{eq:dual1}).}
\end{align*}

Recall that $\Delta P$ is equal to $|\ptnfeas|-|\ptnstart|$ minus the number of pairs added to \pairslist{} in the current iteration. 
Hence, to conclude that $\Delta P \le 2\Delta D$, we need to show that if $t=0$, then a pair is added to $\pairslist{}$ by {\sc Find-Merge-Pair}.

We will say that a component $A$ is able to \emph{reach} $\nodeinTtwo{u}$ if $\nodeinTtwo{u}\in V_2[A]$ or if $\lca_2(A)\prec \nodeinTtwo{u}$ and all intermediate nodes on the path from $\lca_2(A)$ to $\nodeinTtwo{u}$ are not covered by any component in $\ptnfeas$.
The following lemma (which is actually valid in general, and not only for Case 1) enumerates precisely the situations when a merge is possible.
\begin{lemma}\label{lem:possiblemerges}
    Let $A_0 \in \ptnstart$, and let $\ptncur$ denote the set of components in $\ptnfeas$ that are contained in $A_0$.
    Then there exist a pair of elements in $A_0$ that can be added to \pairslist{} if and only if at least one of the following is true:
\begin{enumerate}[(a)]
\item
    $\ptncur$ contains a \bicolored{} component.
\item
    There is a node $\nodeinTtwo{u}\in V_2$ that can be reached by two red components or two blue components in $\ptncur$.
\item
    There is a node $\nodeinTtwo{u}\in V_2$ that can be reached by a red and a blue component in $\ptncur$, but is not covered by these components. 
    Furthermore, the node $\nodeinTtwo{u}$ must satisfy that the nodes on the path from $\nodeinTtwo{u}$ to $\lca_2(A_0)$ are not covered by any red or blue component in $\ptncur$.  
\end{enumerate}
\end{lemma}
\begin{proof}
    By the definitions of {\sc Split} and {\sc Special-Split}, a \multicolored{} component in $\ptnfeas$ must be a top component with blue and white leaves created by applying {\sc Special-Split} to a \tricolored{} component.
    By Lemma~\ref{lem:feas}, there is at most one such \multicolored{} component; if it exists, call it $A^*$.
\begin{enumerate}[(a)]
    \item If $A^*$ exists, then let $A\cup A^*$ be the \tricolored{} component from which {\sc Special-Split} formed a red component $A$ and the \bicolored{} component $A^*$. Then we can merge $A$ and $A^*$ to obtain a new partition that is \feas{$R\cup B$}: it is clear that undoing the {\sc Special-Split} yields a partition that does not overlap any other component of $\ptn$ in $V_2$. The new component $A\cup A^*$ is \comp{$R\cup B\cup \{w\}$} since $A\cup A^*$ is \RBcomp{} and every \tricolored{} triple in $A\cup A^*$ is compatible.
Since the new unique \bicolored{} component is a top component, by Lemma~\ref{lem:ugh},  the new partition also does not overlap in $V_1[R\cup B]$.

	\item If $\ptncur$ does not contain a \bicolored{} component $A^*$, suppose $A, A' \in \ptncur$ are distinct red components in $\ptncur$ so that $A$ and $A'$ can both reach the same node $\nodeinTtwo{u}$ in $V_2$. Then merging $A$ and $A'$ gives a new partition that does not overlap in $V_2$, and which has no \multicolored{} components.  Since $A_0 \cap R$ is compatible, so is $A \cup A'$.  By Lemma~\ref{lem:ugh}  the new partition does not overlap in $V_1[R\cup B]$. Hence, merging $A$ and $A'$ gives a new partition that is \feas{$R\cup B$}.
    
        The same applies if $A$ and $A'$ are both blue components in $\ptncur$.

    \item If $\ptncur$ does not contain a \bicolored{} component $A^*$, suppose there exist $A, A' \in \ptncur$ with $A$ red and $A'$ blue such that 			\begin{inparaenum}[(i)] 
    \item there exists $\nodeinTtwo{u}\in V_2 \setminus (V_2[A] \cup V_2[A'])$ that can be reached by both $A$ and $A'$; and
    \item the nodes on the path from $\nodeinTtwo{u}$ to $\lca_2(A_0)$ are not in $V_2[A'']$ for any $A''\not \subseteq W$.
        \end{inparaenum}      
	Then merging $A$ and $A'$ gives a new partition that does not overlap in $V_2$ and the new component $A\cup A'$, is \RBcomp{} by (i). Thus the new partition is \comp{$R\cup B\cup \{w\}$} for any $w\in \leaves$. 
        The unique  \bicolored{} component $A\cup A'$ in this new partition satisfies that any node on the path from $\nodeinTtwo{u}=\lca_2(A\cup A')$ to $\lca_2(A_0)$ is not covered by a component that is not white. The components of the partition that overlap a node on the path from $\lca_2(A_0)$ to $\lca_2(R\cup B)$ were not changed in the current iteration. Therefore, by Lemma~\ref{lem:ugh}, the new partition does not overlap in $V_1[R\cup B]$. 
     Hence, merging $A$ and $A'$ gives a new partition that is \feas{$R\cup B$}.    
\end{enumerate}
We note that the above three cases encompass all possible merge opportunities within $\ptncur$.
If two  components cannot reach the same node $\nodeinTtwo{u}\in V_2$, then merging them gives a partition that overlaps in $V_2$. If a red and blue component $A$ and $A'$ can reach the same node $\nodeinTtwo{u}\in V_2$ and this node is covered by either $A$ or $A'$, then $A\cup A'$ is not \RBcomp. 
And if a red and blue component $A$ and $A'$ can reach a node $\nodeinTtwo{u}\in V_2$ that is not in $V_2[A] \cup V_2[A']$, but some node on the path from $\nodeinTtwo{u}$ to $\lca_2(A_0)$ is covered by a component $A''\in \ptncur$ that is red or blue, then $A\cup A'$ will overlap $A''$ in $V_1[R]$ or $V_1[B]$. To see this, assume $A''$ is red (the blue case is analogous) and let $\nodeinTtwo{v}$ be the node in $V_2[A'']$ closest to $\nodeinTtwo{u}$ on the path from $\nodeinTtwo{u}$ to $\lca_2(A_0)$. Then $\nodeinTtwo{v}=\lca_2(A\cup (A''\cap\lbelow{\nodeinTtwo{v}}) \prec \lca_2(A'')$, and since $A\cup A''$ are compatible in $R$, we should also have $\lca_1(A\cup (A''\cap\lbelow{\nodeinTtwo{v}}) \prec \lca_1(A'')$. Thus $A''$ and $A$ overlap on  a node on the path from $\lca_1(A)$ to $\lca_1(R\cup B)$.
\end{proof}

We are now ready to complete the proof of Proposition~\ref{prop:1}, by showing that if $t=0$, then at least one of (a), (b) and (c) in Lemma~\ref{lem:possiblemerges} holds for $\ptnfeas$.

Suppose (a) does not hold, i.e., $\ptnfeas$ has only \monochrome{} components. 
Let $\nodeinTtwo{u}$ be the last node chosen in {\sc Make-\RBcomp} to subdivide the top component. The existence of $\nodeinTtwo{u}$ follows since $\ptnRBcomp \neq \ptnstart$ as we observed above.
Let $A\subset \lbelow{\nodeinTtwo{u}}$ be the non-top component added to the partition at this point. Since $t=0$, we have that $A$ is \bicolored{}  by property~\ref{itm:multi1} in Lemma~\ref{lem:properties}. {\sc Split} will  split  $A$ into $A\cap R$ and $A\cap B$, and by property~\ref{itm:binotop} in Lemma~\ref{lem:properties}, node $\nodeinTtwo{u}$ itself is not covered by any component of $\ptnfeas$ and it can be reached by red component $A\cap R$ and blue component $A\cap B$. So if there is no node on the path in $T_2$ from $\nodeinTtwo{u}$ to $\lca_2(A_0)$  that is covered by a red or blue component in $\ptnfeas$, then $A\cap R$ and $A\cap B$ give a pair that can be added to \pairslist{} according to (c). 

So suppose that  there is a node on the path from $\nodeinTtwo{u}$ to $\lca_2(A_0)$ that is covered by a (say) red component $A_R$ in $\ptnfeas$; without loss of generality $\nodeinTtwo{v}$ is the node closest to $\nodeinTtwo{u}$ such that $\nodeinTtwo{v}\in V_2[A_R]$ for some red component $A_R$. If $A\cap R$ can reach $\nodeinTtwo{v}$ then $\nodeinTtwo{v}$ can  be reached by two red components, so we can add a pair  to \pairslist{} according to (b).

Otherwise, let $\nodeinTtwo{w}$ be the node closest to $\nodeinTtwo{v}$ on the path from $\nodeinTtwo{u}$ to $\nodeinTtwo{v}$ that is covered by a white component $A_W$. By definition of $\nodeinTtwo{w}$, all nodes on the path from $\nodeinTtwo{w}$ to $\nodeinTtwo{v}$ are not covered by any component in $\ptnfeas$. 
Observe that 
\begin{itemize}
\item
$A_W$ and $A_R$ must have been part of the top component of $\ptnRBcomp$, by definition of $\nodeinTtwo{u}$. 
\item
$A_W$ and $A_R$ cannot have been part of the top component of $\ptnsplit$: as we observed above, the top component of $\ptnsplit$ contains a white leaf $\whiteleaf$ that is not a descendant of $\lca_2(A_0\cap (R\cup B))$, so $A_W\cup \{\whiteleaf\}$ overlaps $\nodeinTtwo{v}$, and since $A_R$ also overlaps $\nodeinTtwo{v}$, the top component of $\ptnsplit$ would not be \splittable{} if it contained $A_R\cup A_W$, contradicting Lemma~\ref{lem:splittable}.
\item
$A_W$ was part of non-top component $A''$ in $\ptnsplit$, which contained red and white leaves:
by the second observation, $A_W$ must have been part of a non-top component $A''$, which was \multicolored{} by property~\ref{itm:multi1} of Lemma~\ref{lem:properties}. Since $A''$ and $A_R$ were part of the same component of $\ptnRBcomp$ by the first observation, $A''\cup A_R$ must be \RBcomp{} by Lemma~\ref{lem:RBcomp}. This implies $A''$ does not contain any blue leaves, since otherwise such a blue leaf $\blueleaf{}$, and a red leaf $\redleaf{}\in A_R\cap \lbelow{\nodeinTtwo{v}}$ and a red leaf $\altredleaf{}\in A_R\setminus \lbelow{\nodeinTtwo{v}}$ would have $\lca_2(\blueleaf{},\redleaf{}) = \nodeinTtwo{v} \prec \lca_2(\blueleaf{}, \redleaf{}, \altredleaf{})$, and thus the triple would be in\compatible{}, contradicting that $A''\cup A_R$ is \RBcomp{}. 
\end{itemize}
Thus, $A''= A_W \cup (A''\cap R)$, and by property~\ref{itm:binotop} in Lemma~\ref{lem:properties}, the component $A''\cap R$ in $\ptnfeas$ can reach $\lca_2(A'')$. Since $\nodeinTtwo{w}\prec \lca_2(A'')$,  $\lca_2(A'')$ is on the path from $\nodeinTtwo{w}$ to $\nodeinTtwo{v}$. But then $A''\cap R$ can reach $\nodeinTtwo{v}$,  so $A_R$ and $A''\cap R$ give a pair to be added to \pairslist{} according to (b).
\end{proof}

\section{A compact formulation of the LP}\label{sec:compact}

Here we give a compact formulation for \eqref{eq:lp}.
This shows that it can be optimized efficiently.
While this is not needed in our algorithm, it is possible that an LP-rounding based algorithm could achieve a better approximation guarantee, in which case this formulation will be of use.
Moreover, the LP explicitly encodes the structure of compatible sets in a way that \eqref{eq:lp} does not; we believe this may provide additional structural insights in the future.

\medskip

We remark that \eqref{eq:lp} can also be shown to be polynomially solvable by providing a separation oracle for the dual. 
The dual of \eqref{eq:lp} is similar to \eqref{eq:dual}, the dual of \eqref{eq:lpbig}, except that $z$ is indexed only by singletons and not arbitrary subsets of $\leaves$. 
This dual has a polynomial number of variables, but an exponential number of constraints. 
By the equivalence of separation and optimization, it suffices to provide a separation oracle for this dual. 
The following problem subsumes this separation problem: given some (positive or negative) weights $y$ on the nodes of $V$, find a compatible subset $L \in \compat$ which maximizes $\sum_{v \in V[L]} y_v$.
This is a weighted variant of the \emph{maximum agreement subtree problem}. 
    Similar to the usual (unweighted) version~\cite{Steel93}, this can be solved in polynomial time via dynamic programming.

\bigskip
\newcommand{\W}{Z}
\newcommand{\arcsL}{U_1}
\newcommand{\arcsR}{U_2}

Assume for convenience that $\leaves = \{1,2,\ldots, n\}$.
We will deviate from the notational conventions in the previous sections, and use $i$ and $j$ to denote leaves, and we will use $t\in\{1,2\}$ to denote the indices of the two input trees. 

Consider a set $L \subseteq \leaves$, and the tree $T_t[L]$ it induces in $T_t$ for $t \in \{1,2\}$.
We can identify each internal node $u$ of $T_t$ that has two children in $T_t[L]$ by a pair consisting of the smallest leaf in $L$ in the subtree below each of the two children. If these two leaves are $i$ and $j$, where $i\leq j$, we will label $u$ with $(i,j)$. Note that $u = \lca_t(i,j)$. We extend this labeling to the leaves in $L$, and label $i\in L$ with the pair $(i,i)$. 

The set $L$ is compatible precisely if the set of labels of the labelled nodes of $T_1[L]$ and $T_2[L]$ are the same, and the ancestry relationship between the labelled nodes are also the same.
Observe that if an internal node $u$ in $T_t[L]$ is labelled with $(i_1,i_2)$, then it must have two descendants $u_1$ and $u_2$ where $u_1$ is labelled $(i_1,j)$ and $u_2$ is labelled $(i_2,k)$ for some $j,k\in L$.

Using this observation, we can use the labelling to describe compatibility constraints.
To do this, we begin by constructing a directed acyclic graph $D=(\W,\arcsL \cup \arcsR)$ as follows.
The node set $\W$ consists of all pairs $(i_1,i_2) \in \leaves^2$ for which $i_1 \leq i_2$.
With a slight abuse of notation, we define $\lca_t(r)$ for $r=(i_1,i_2)\in \W$ as $\lca_t(i_1,i_2)$. 
Given two nodes $r = (i_1, i_2)$ and $s = (j_1, j_2)$ in $\W$:
\begin{itemize}
    \item if $\lca_t(s) \prec \lca_t(r)$ for all $t\in \{1,2\}$ and $i_1 = j_1$, then $(r,s) \in \arcsL$;
    \item if $\lca_t(s) \prec \lca_t(r)$ for all $t\in \{1,2\}$ and $i_2 = j_1$, then $(r,s) \in \arcsR$.
\end{itemize}
Suppose $(r,r_\Lc)\in \arcsL, (r,r_\Rc)\in \arcsR$. Observe that then $\lbelow{\lca_t(r_\Lc)}$ and $\lbelow{\lca_t(r_\Rc)}$ are disjoint subsets of $\lbelow{\lca_t(r)}$ for $t\in \{1,2\}$. This implies that $r_\Lc$ and $r_\Rc$ cannot both have a directed path to the same node \modified{$s=(j_1,j_2)$}, since that would imply that $j_1,j_2\in \lbelow{\lca_t(r_\Lc)}\cap \lbelow{\lca_t(r_\Rc)}$.

Define $\W_\leaves = \{ (i,i) : i \in \leaves\}$.
Let $\mathcal{F}$ denote the set of out-arborescences in $D$ with leaf set contained in $\W_\leaves$ and where each internal node has one outgoing arc in $\arcsL$ and one outgoing arc in $\arcsR$.
Then the above implies that $L \in \compat$ if and only if there is an $F(L) \in \mathcal{F}$ with leaf set corresponding to $L$. (If $|L| = 1$, $F(L)$ is empty.)
Let $\chi_F \in \{0,1\}^{\arcsL \cup \arcsR}$ be the characteristic vector of the arc set of $F$, for any $F \in \mathcal{F}$. Let $\coneT$ denote the cone generated by $\{\chi_F : F \in \mathcal{F}\}$, i.e., $y\in \coneT$, if and only if there exists $x_L\ge 0$ for $L\in {\cal C}$ such that $y=\sum_{L \in \compat: |L| \geq 2} x_L \chi_{F(L)}$.

We begin by giving a description of $\coneT$. 
For $r \in \W$, let $\delta^+(r)$ denote the arcs in $D$ leaving $r$, and $\delta^-(r)$ the arcs entering $r$.
For $S \subseteq \arcsL \cup \arcsR$, let $y(S) = \sum_{a \in S} y_a$.

\begin{lemma}\label{lem:cone}
    \begin{align*}
        \coneT = \{ y \in \R_+^{\arcsL \cup \arcsR} : 
            y(\delta^+(r) \cap \arcsL) &= y(\delta^+(r) \cap \arcsR) \qquad \forall r \in \W  \setminus \W_\leaves \\
    y(\delta^+(r) \cap \arcsL) &\geq y(\delta^-(r)) \qquad \qquad \;\;\forall r \in \W \setminus \W_\leaves \}.
\end{align*}
\end{lemma}
\begin{proof}
    Let $Y$ denote the cone described by the right hand side of the claimed equality.
    It is clear that $\chi_F \in Y$ for any $F \in \mathcal{F}$, and hence that $\coneT \subseteq Y$.
    It remains to show that $Y \subseteq \coneT$.

    Suppose $y \in Y$; we prove that $y \in \coneT$, proceeding by induction on the size of the support of $y$.
    The claim trivially holds if $y=0$. So suppose $y\neq 0$.
    Choose $r=(i_1,i_2) \in \W$ such that $y(\delta^-(r)) = 0$ but $y(\delta^+(r)) > 0$.
    We now find an arborescence $F \in \mathcal{F}$ rooted at $r$ and contained in the support of $y$.
    This is trivial if $i_1 = i_2$; if not, we proceed as follows.
    Choose any $(r,r_\Lc) \in \arcsL \cap \delta^+(r)$ and $(r, r_\Rc) \in \arcsR \cap \delta^+(r)$ that are both in the support of $y$.
    Arguing inductively, we obtain arborescences $F_\Lc$ and $F_R$ in the support of $y$ rooted at $r_\Lc$ and $r_\Rc$ respectively.
    We have already noted that there is no node that both $r_\Lc$ and $r_\Rc$ can reach; thus $F_\Lc$ and $F_R$ are disjoint.
    We obtain $F$ by combining $F_\Lc$, $F_R$ and the arcs from $r$.

    Now set $y'= y - \epsilon \chi_F$, where $\epsilon$ is chosen maximally so that $y' \geq 0$.
    So $y'$ has smaller support, and so by induction, $y' \in \coneT$.
    Hence $y = y' + \epsilon \chi_F$ is too.
\end{proof}

\newcommand{\rootof}[1]{\rt(#1)}
\newcommand{\cLPtag}{$\text{LP}^\star$}

Using Lemma~\ref{lem:cone}, we now describe our compact formulation. 
For $t \in \{1,2\}$ and $v \in V_t$, let $\lca^{-1}(v) = \{ r \in \W : \lca_t(r) = v \}$.
\begin{align}
    \min \qquad \sum_{r \in \W \setminus \W_\leaves} &\bigl(y(\delta^+(r) \cap \arcsL) - y(\delta^-(r))\bigr) \;+\; \sum_{i \in \leaves} x_{\{i\}} \;-\;1 \hspace{-6cm} \tag{\cLPtag-1}\label{eq:compobj}\\
    \text{s.t.} 
	\qquad \qquad \qquad \qquad y(\delta^+(r) \cap \arcsL) &= y(\delta^+(r) \cap \arcsR) &&\forall r \in \W \tag{\cLPtag-2}\label{eq:conea}\\
    y(\delta^+(r) \cap \arcsL) &\geq y (\delta^-(r)) &&\forall r \in \W \tag{\cLPtag-3}\label{eq:coneb}\\
    x_{\{i\}} &= 1 - y(\delta^-(i,i)) &&\forall i \in \W  \tag{\cLPtag-4}\label{eq:leafsat}\\
    \sum_{r \in \lca^{-1}(v)} y(\delta^+(r) \cap \arcsL) &\leq 1 && \forall v \in V \setminus \leaves \tag{\cLPtag-5} \label{eq:overlapsat}\\
    x_{\{i\}} &\geq 0 && \forall i \in \leaves \notag\\
    y_a  &\geq 0 && \forall a \in \arcsL \cup \arcsR \notag
\end{align}

\begin{lemma}
    This LP is equivalent to \eqref{eq:lp}.
\end{lemma}
\begin{proof}
By Lemma~\ref{lem:cone}, \eqref{eq:conea} and \eqref{eq:coneb} ensure $y\in \coneT$ and thus
we can expand $y = \sum_{L \in \compat: |L| \geq 2} x_L \chi_{F(L)}$ for $x \geq 0$.
Then taking $x_L$ for $|L| = 1$ as defined by \eqref{eq:leafsat}, we ensure that $\sum_{L \in \compat: i \in L} x_L = 1$ for all $i \in \leaves$.
Constraint \eqref{eq:overlapsat} becomes $\sum_{L \in \compat: v \in L} x_L \leq 1$ for all non-leaves $v$.
And for $r \in \W \setminus \W_\leaves$, 
\[ y(\delta^+(r) \cap \arcsL) - y(\delta^-(r)) = \sum_{L \in \compat: \rootof{F(L)} = r} x_L, \]
where $\rootof{F}$ denotes the root of $F$.
So the objective becomes $\sum_{L \in \compat} x_L \;-\;1$.
Thus, this formulation is equivalent to \eqref{eq:lp}.
\end{proof}

\paragraph{Acknowledgements.}
    We acknowledge the support of the Tinbergen Institute and
    the Hausdorff Research Institute for Mathematics, where portions of this research was pursued.

N.O. was supported in part by NWO Veni grant 639.071.307.
F.S. was supported in part by NSF grants CCF-1526067 and CCF-1522054.
A.Z. was supported in part by grant  \#359525 from the Simons Foundation.

\bibliography{maf}

\begin{thebibliography}{10}

\bibitem{Allen2001}
B.~L. Allen and M.~Steel.
\newblock Subtree transfer operations and their induced metrics on evolutionary
  trees.
\newblock {\em Annals of Combinatorics}, 5(1):1--15, 2001.

\bibitem{LCARevisited}
M.~A. Bender and M.~Farach-Colton.
\newblock The {LCA} problem revisited.
\newblock In {\em Proceedings of the 4th Latin American Symposium on
  Theoretical Informatics (LATIN)}, pages 88--94, 2000.

\bibitem{Bonet06}
M.~L. Bonet, K.~S. John, R.~Mahindru, and N.~Amenta.
\newblock Approximating subtree distances between phylogenies.
\newblock {\em Journal of Computational Biology}, 13(8):1419--1434, 2006.

\bibitem{Bordewich08}
M.~Bordewich, C.~McCartin, and C.~Semple.
\newblock A 3-approximation algorithm for the subtree distance between
  phylogenies.
\newblock {\em Journal of Discrete Algorithms}, 6(3):458--471, 2008.

\bibitem{Bordewich04}
M.~Bordewich and C.~Semple.
\newblock On the computational complexity of the rooted subtree prune and
  regraft distance.
\newblock {\em Annals of Combinatorics}, 8(4):409--423, 2004.

\bibitem{chataigner2005approximating}
F.~Chataigner.
\newblock Approximating the maximum agreement forest on k trees.
\newblock {\em Information processing letters}, 93(5):239--244, 2005.

\bibitem{chen2016approximating}
J.~Chen, F.~Shi, and J.~Wang.
\newblock Approximating maximum agreement forest on multiple binary trees.
\newblock {\em Algorithmica}, 76(4):867--889, 2016.

\bibitem{chen2017cubic}
Z.-Z. Chen, Y.~Harada, and L.~Wang.
\newblock A new 2-approximation algorithm for {rSPR} distance.
\newblock In Z.~Cai, O.~Daescu, and M.~Li, editors, {\em Bioinformatics
  Research and Applications}, pages 128--139, Cham, 2017. Springer
  International Publishing.

\bibitem{chen2016improved}
Z.-Z. Chen, E.~Machida, and L.~Wang.
\newblock An improved approximation algorithm for {rSPR} distance.
\newblock In {\em International Computing and Combinatorics Conference}, pages
  468--479. Springer, 2016.

\bibitem{Darwin37}
C.~Darwin.
\newblock {\em Notebook B: Transmutation of species (1837?-1838)}.
\newblock In: John van Wyhe: The Complete Work of Charles Darwin Online, 2002.
\newblock \url{http://darwin-online.org.uk/}.

\bibitem{MathEvPhyl}
O.~Gascuel, editor.
\newblock {\em Mathematics of Evolution and Phylogeny}.
\newblock Oxford University Press, Inc., 2005.

\bibitem{GoemansW97}
M.~X. Goemans and D.~P. Williamson.
\newblock The primal-dual method for approximation algorithms and its
  application to network design problems.
\newblock In D.~S. Hochbaum, editor, {\em Approximation Algorithms for
  {NP}-hard Problems}, pages 144--191. PWS Publishing Co., Boston, 1997.

\bibitem{Harel80}
D.~Harel.
\newblock A linear time algorithm for the lowest common ancestors problem.
\newblock In {\em Proceedings of the 21st Annual Symposium on Foundations of
  Computer Science (FOCS)}, pages 308--319, 1980.

\bibitem{HarelTarjan}
D.~Harel and R.~E. Tarjan.
\newblock Fast algorithms for finding nearest common ancestors.
\newblock {\em SIAM Journal on Computing}, 13(2):338--355, 1984.

\bibitem{Hein96}
J.~Hein, T.~Jiang, L.~Wang, and K.~Zhang.
\newblock On the complexity of comparing evolutionary trees.
\newblock {\em Discrete Applied Mathematics. The Journal of Combinatorial
  Algorithms, Informatics and Computational Sciences}, 71(1-3):153--169, 1996.

\bibitem{HusonRuppScornavacca10}
D.~Huson, R.~Rupp, and C.~Scornavacca.
\newblock {\em Phylogenetic Networks: Concepts, Algorithms and Applications}.
\newblock Cambridge University Press, 2010.

\bibitem{Nakhleh2009ProbSolv}
L.~Nakhleh.
\newblock Evolutionary phylogenetic networks: models and issues.
\newblock In L.~Heath and N.~Ramakrishnan, editors, {\em The Problem Solving
  Handbook for Computational Biology and Bioinformatics}. Springer, 2009.

\bibitem{implementation}
N.~Olver, F.~Schalekamp, L.~Stougie, and A.~van Zuylen.
\newblock Implementation of the {MAF} algorithm and compact formulation.
\newblock Available at \url{http://nolver.net/maf} and
  \url{http://frans.us.MAF}, 2018.

\bibitem{Rodrigues03}
E.~M. Rodrigues.
\newblock {\em Algoritmos para Compara\c c\~ao de \'Arvores Filogen\'eticas e o
  Problema dos Pontos de Recombina\c c\~ao}.
\newblock PhD thesis, University of S\~ao Paulo, Brazil, 2003.
\newblock Chapter 7, available at
  \url{http://www.ime.usp.br/~estela/studies/tese-traducao-cp7.ps.gz}.

\bibitem{Rodrigues07}
E.~M. Rodrigues, M.-F. Sagot, and Y.~Wakabayashi.
\newblock The maximum agreement forest problem: approximation algorithms and
  computational experiments.
\newblock {\em Theoretical Computer Science}, 374(1-3):91--110, 2007.

\bibitem{DBLP:conf/icalp/SchalekampZS16}
F.~Schalekamp, A.~van Zuylen, and S.~van~der Ster.
\newblock A duality based 2-approximation algorithm for maximum agreement
  forest.
\newblock In {\em Proceedings of the 43rd International Colloquium on Automata,
  Languages, and Programming (ICALP)}, volume~55 of {\em LIPIcs}, pages
  70:1--70:14. Leibniz-Zentrum f\"ur Informatik, 2016.

\bibitem{SempleSteel2003}
C.~Semple and M.~Steel.
\newblock {\em Phylogenetics}.
\newblock Oxford University Press, 2003.

\bibitem{Shietal15}
F.~Shi, Q.~Feng, J.~You, and J.~Wang.
\newblock Improved approximation algorithm for maximum agreement forest of two
  rooted binary phylogenetic trees.
\newblock {\em Journal of Combinatorial Optimization}, 2015.

\bibitem{Steel93}
M.~Steel and T.~Warnow.
\newblock Kaikoura tree theorems: Computing the maximum agreement subtree.
\newblock {\em Information Processing Letters}, 48(2):77--82, Nov. 1993.

\bibitem{vanIersel14}
L.~van Iersel, S.~Kelk, N.~Lekic, and L.~Stougie.
\newblock Approximation algorithms for nonbinary agreement forests.
\newblock {\em SIAM Journal on Discrete Mathematics}, 28(1):49--66, 2014.

\bibitem{Whidden13}
C.~Whidden, R.~G. Beiko, and N.~Zeh.
\newblock Fixed-parameter algorithms for maximum agreement forests.
\newblock {\em SIAM Journal on Computing}, 42(4):1431--1466, 2013.

\bibitem{Whidden09}
C.~Whidden and N.~Zeh.
\newblock A unifying view on approximation and {FPT} of agreement forests.
\newblock In {\em Algorithms in Bioinformatics}, volume 5724 of {\em Lecture
  Notes in Computer Science}, pages 390--402. Springer Berlin Heidelberg, 2009.

\bibitem{Wu09}
Y.~Wu.
\newblock A practical method for exact computation of subtree prune and regraft
  distance.
\newblock {\em Bioinformatics}, 25(2):190--196, 2009.

\bibitem{Wu10}
Y.~Wu and J.~Wang.
\newblock Fast computation of the exact hybridization number of two
  phylogenetic trees.
\newblock In {\em Bioinformatics Research and Applications}, volume 6053 of
  {\em Lecture Notes in Computer Science}, pages 203--214. Springer Berlin
  Heidelberg, 2010.

\end{thebibliography}

\appendix
\section{The running time}\label{sec:running}

It is quite clear from the definition of the Red-Blue Algorithm that it runs in polynomial time.
In this section we show that it can be implementated to run in $O(n^2)$ time, where $n$ denotes the number of leaves.
(We work in the random access machine model of computation, and assume a word size of $\Omega(\log n)$.)

We note that our presentation is focused on showing the bound on the running time as straightforwardly as possible, and there are some places where a more careful implementation is more efficient. However, we have not been able to find an implementation with an overall running time of $o(n^2)$. %

\newcommand{\cp}[1]{\nodeinTtwo{p}_{#1}} %
\newcommand{\bsz}{s}
\newcommand{\csz}{s} 
\newcommand{\lbl}{\text{m}}

We assume $\ptn$ is a given partition (not overlapping in $V_2$), that is stored such that we can query the size of any component in constant time, and that for each node $\nodeinTtwo{u}\in V_2$, we can query $A_{\nodeinTtwo{u}}$, the component in $\ptn$ that covers $\nodeinTtwo{u}$ (which will be equal to $\emptyset$ if $\ptn$ does not cover $\nodeinTtwo{u}$), and $\csz(\nodeinTtwo{u}) = |\lbelow{\nodeinTtwo{u}}\cap A_{\nodeinTtwo{u}}|$.
Note that we can determine this information by a bottom-up pass of $T_2$ in $O(n)$ time. We will recompute it whenever we refine $\ptn$; since there can only be at most $n-1$ refinement operations, the total time to maintain this information is $O(n^2)$.

By Harel~\cite{Harel80} (see also~\cite{HarelTarjan, LCARevisited}), we furthermore may assume that the computation of $\lca_i(u,v)$ for given nodes $u,v\in V_i$ takes constant time (after a linear preprocessing time). 
It immediately follows from this that we can determine whether or not $u \preceq v$ in tree $T_i$ in constant time as well.

We will show that the time between subsequent refinements of $\ptn$ is $O(n)$. This bounds the time of the main loop of the algorithm by $O(n^2)$. The only remaining part of the algorithm is the {\sc Merge-Components} step, which will perform at most $n-1$ merges, each of which can clearly be done in $O(n)$ time.

\subsubsection*{\it Finding a lowest \pcs{}}
We make a single pass through $T_1$, in bottom-up order (starting from the leaves), until we find a \pcs{}. We will spend constant time per node, thus showing that the time to find a lowest \pcs{} is $O(n)$.

For each node $u\in V_1$ that we have already considered, $A_u$ references the component $A \in \ptn$ which covers $u$, with $A_u = \emptyset$ if there are no such components. (If there are multiple such components, $u$ is a \pcs.) Furthermore, $\cp{u}$ is equal to $\lca_2(A_u \cap \lbelow{u})$, and $\bsz(u)$ is the size of $A_u \cap \lbelow{u}$.
Observe that for any $x \in \leaves$, we know $A_x$, the component containing $x$, and $\bsz(x) = 1$ and $\cp{x} = x$.

Given a non-leaf node $u \in V_1$, with children $u_1$ and $u_2$ that have already been considered, we can determine whether $u$ is a \pcs{}, and if not determine $A_u, \cp{u}$ and $\bsz(u)$, in constant time:
If either (or both) of $A_{u_1}$ and $A_{u_2}$ do not cover $u$ (which can be determined by checking if $A_{u_i} = \emptyset$ or $\bsz(u_i) = |A_{u_i}|$),
    set all the values according to which child (if any) does cover $u$, and end the consideration of node $u$.
    So assume from now on that both do cover $u$.
    
    If $A_{u_1} \neq A_{u_2}$, then $u$ satisfies the second condition of a \pcs{}, and we are done.
        Otherwise, $A_u = A_{u_1} = A_{u_2}$.
    Set $\cp{u} = \lca_2(\cp{u_1}, \cp{u_2})$ and  $\bsz(u) = \bsz(u_1) + \bsz(u_2)$.
        If $\cp{u_1} \nprec \cp{u}$ or $\cp{u_2} \nprec \cp{u}$, then $\lbelow{u}$ is incompatible, and $u$ satisfies the first condition of a \pcs{}.
    If $\csz(\cp{u}) = |A_u|$ and $\bsz(u) < |A_u|$ then $u$ satisfies the third condition for being a \pcs{}: by $\bsz(u) < |A_u|$, we know $A_u\setminus\lbelow{u}\neq\emptyset$. For any $w\in A_u\setminus\lbelow{u}$,  $\lca_1( A_u \cap \lbelow{u} ) \preceq u \preceq \lca_1( A_u \cap \lbelow{u} \cup \{ w \} )$, while by $\csz(\cp{u})=|A_u|$ we know that $\lca_2( A_u \cap \lbelow{u} ) = \lca_2( A_u )$. So $A_u\cap \lbelow{u}\cup \{w\}$ is in\compatible{} for any $w\in A_u\setminus \lbelow{u}$.
Otherwise, $u$ is not a \pcs{}, and we  finish our consideration of $u$.

\medskip

Once we have determined the coloring $(R, B, W)$, we compute $|A\cap C|$ for each component $A\in \ptn$ and $C\in \{R, B, W\}$. We also compute three additional labels for each node $\nodeinTtwo{u}\in V_2$: $\csz_C(\nodeinTtwo{u}) = |A_{\nodeinTtwo{u}} \cap \lbelow{\nodeinTtwo{u}}\cap C|$ for $C\in \{R, B, W\}$.
This information can be determined by a bottom-up traversal of $T_2$ in $O(n)$ time. We assume this information is updated whenever the partition is refined.

\subsubsection*{\sc Make-\RBcomp}
Consider the nodes of $T_2$ in bottom-up order, until we find a node $\nodeinTtwo{u}$ such that both $\csz_B(\nodeinTtwo{u})>1$ and $\csz_R(\nodeinTtwo{u})>1$. Since $\nodeinTtwo{u}$ is a lowest such node, if $|A_{\nodeinTtwo{u}}\cap R| = \csz_R(\nodeinTtwo{u})$ and $|A_{\nodeinTtwo{u}}\cap B|= \csz_B(\nodeinTtwo{u})$, then $A_{\nodeinTtwo{u}}$ is \RBcomp{}; otherwise $\nodeinTtwo{u}$ is precisely as indicated in {\sc Make-\RBcomp}.

\subsubsection*{\sc Make-Splittable}
We again consider the nodes in $T_2$ in bottom-up order.
For any node $\nodeinTtwo{u}$ with $A_{\nodeinTtwo{u}} \neq \emptyset$, using $\csz_C(\nodeinTtwo{u})$ for $C\in \{R, B, W\}$, we can check in $O(1)$ time whether $A_{\nodeinTtwo{u}} \cap \lbelow{\nodeinTtwo{u}}$ is \bicolored, and that for any $C \in \{R,B,W\}$ with $A_{\nodeinTtwo{u}} \cap C \neq \emptyset$, that $\csz_C(\nodeinTtwo{u}) < |A_{\nodeinTtwo{u}} \cap C|$ (and hence $(A_{\nodeinTtwo{u}} \setminus \lbelow{\nodeinTtwo{u}}) \cap C \neq \emptyset$).

\subsubsection*{\it \bf\sc Split}
Note that a regular split of a component $A$ can be done in $O(n)$ time, by simply checking the color of each leaf in $A$ and partitioning $A$ accordingly.
We now show how to check if $A$ needs a {\sc Special-Split} (and if so which of the two possible refinements is
applied) by considering the nodes in $V_2[A]$ in bottom-up order.

If $A$ is \tricolored{}, then the fact that $A$ is \RBcomp{} and \splittable{} implies that there exist $\nodeinTtwo{u}_R$ and $\nodeinTtwo{u}_B$ that are covered by $A$ and for which $A\cap \lbelow{\nodeinTtwo{u}_R}= A\cap R$ and $A\cap \lbelow{\nodeinTtwo{u}_B}=A\cap B$. Using a bottom-up traversal of $V_2$ will find $\nodeinTtwo{u}_R$ and $\nodeinTtwo{u}_B$ (they are the first nodes $\nodeinTtwo{v}$ encountered such that $A_{\nodeinTtwo{v}}=A$ and $\csz_C(\nodeinTtwo{v})=|A\cap C|$ for $C=R$ and $B$ respectively).

Given $\nodeinTtwo{u}_R$ and $\nodeinTtwo{u}_B$, we check if a {\sc Special-Split} is required, by considering $\nodeinTtwo{u} = \lca_2(A\cap (R\cup B))=\lca_2(\nodeinTtwo{u}_R,\nodeinTtwo{u}_B)$;  a {\sc Special-Split} is required exactly if $\csz(\lca_2(\nodeinTtwo{u}_R,\nodeinTtwo{u}_B))<|A|$, since in that case any $x_W\in A\setminus\lbelow{\lca_2(\nodeinTtwo{u}_R,\nodeinTtwo{u}_B)}$ forms a \compatible{} triple with any $x_R\in A\cap R, x_B\in A\cap B$. If, in addition, $\csz_W(\lca_2(\nodeinTtwo{u}_R,\nodeinTtwo{u}_B))=0$, we know that every \tricolored{} triple in $A$ is \compatible.

\subsubsection*{\it \sc Find-Merge-Pair}
We need to determine if there exist two components (both
intersecting $R \cup B$) that can be merged, in time $O(n)$. If such components are found, then we can take a non-white leaf in each component and add this pair to \pairslist{}.
Recall Lemma~\ref{lem:possiblemerges}, which enumerates all possible situations where a potential merge may exist.
\begin{enumerate}[(a)]
\item
\emph{$\ptn$ has a \bicolored{} component. }
This component must have been created by an application of {\sc Special-Split}, splitting some component $A \cup A^* \in \ptnsplit$ into $A$ and $A^*$.
        As discussed in the proof of Lemma~\ref{lem:possiblemerges}, simply undoing this split is a valid merge.
        Since {\sc Special-Split} is invoked at most once per iteration, we can simply add a pair  to \pairslist{} during {\sc Special-Split}.
\item
\emph{There is a node $\hat u\in V_2$ that can be reached by  two red or two blue components that were part of the same component at the start of the current iteration.} By Lemma~\ref{lem:coloring} (and property~\ref{itm:multi0} of Lemma~\ref{lem:properties}), any two components that are not white that were created in the current iteration must have been part of the same partition at the start  of the iteration. We may assume that we can check for each component in constant time whether it was created in the current iteration.

        We work bottom-up in $T_2$, and set  $\reachu_{\nodeinTtwo{u}}$ to be the set of red and blue components that were created in the current iteration, and that can reach $\nodeinTtwo{u}$ for every $\nodeinTtwo{u}\in V_2$.
    If $\reachu_{\nodeinTtwo{u}}$ contains two components of the same color, these two components can be merged, and we terminate.

Note that if $\reachu_{\nodeinTtwo{u}}$ ever contains three components, we will have found a merge and the algorithm will terminate.  This ensures that we can compute this for $\nodeinTtwo{u}$ in constant time, given the values of $A_{\nodeinTtwo{u}}$ and $\reachu_{\nodeinTtwo{u}_1}$ and $\reachu_{\nodeinTtwo{u}_2}$ for the children of $\nodeinTtwo{u}$.
\item
\emph{There is a node $\nodeinTtwo{u}\in V_2$ that can be reached by a red and a blue component that were part of the same component $A_0$ at the start of the current iteration, but is not covered by these components. 
    Furthermore, the node $\nodeinTtwo{u}$ must satisfy that the nodes on the path from $\nodeinTtwo{u}$ to $\lca_2(A_0)$ are not covered by any red or blue component.} Note that by Lemma~\ref{lem:coloring}, $A_0$ is the only component that was modified in the current iteration, so for this second condition we can simply check that no node on the path from $\nodeinTtwo{u}$ to the root of $V_2$ is covered by a red or blue component that was created in the current iteration.
    
    If we did not find two components of the same color that can be merged, we have found $\reachu_{\nodeinTtwo{u}}$ for every $\nodeinTtwo{u}\in V_2$, where $|\reachu_{\nodeinTtwo{u}}|\le 2$. We now work top-down in $T_2$. If we encounter a node $\nodeinTtwo{u}$ that is covered by a red or blue component that was created in the current iteration, we stop and do not consider the descendants of $\nodeinTtwo{u}$ (since for any such a descendant, $\nodeinTtwo{u}$ is on its path to $\lca_2(A_0)$).
            If we encounter a node $\nodeinTtwo{u}$ such that $|\reachu_{\nodeinTtwo{u}}|=2$ and $\nodeinTtwo{u}$ is not covered by any component, the two components in $\reachu_{\nodeinTtwo{u}}$ can be merged,             and we may terminate.
\end{enumerate}

\section{Integrality gap lower bounds}\label{sec:intgap}

\newcommand{\eps}{\ensuremath{\varepsilon}}
We show a lower bound on the integrality gap of $\tfrac{16}5$ for the integer linear program formulation of Wu~\cite{Wu09}.
Recall that a solution to MAF can be viewed as the leaf sets of the trees in a forest, obtained by deleting edges from the input trees. The formulation has binary variables $x_e$ for every edge $e \in T_1$, indicating whether $e$ is deleted from $T_1$. 
We use $P(i, j)$ to denote the set of edges in $T_1$ on the path between leaves $i$ and $j$, and $P_2(i, j)$ to denote the  set of edges in $T_2$ on the path between leaves $i$ and $j$. 
Wu's integer linear program~\cite{Wu09} is given by:

\newlength{\xlength}
\settowidth{\xlength}{$\displaystyle\sum_{e \in P( i,j )} x_e + \sum_{e \in P( k,\ell )} x_e \geq 1$}
\newlength{\xalength}
\settowidth{\xalength}{$\minimize \quad$}
\newlength{\xxlength}
\setlength{\xxlength}{\textwidth}
\addtolength{\xxlength}{-\xlength}
\addtolength{\xxlength}{-\xalength}
\addtolength{\xxlength}{-1em}

\begin{align*}
\minimize \quad& \sum_e x_e \\
\text{s.t.}\quad & \sum_{e \in P( i,j ) \cup P( i,k ) \cup P( j,k ) } x_e \geq 1 && \mbox{for all incompatible triples $i,j,k$} \\
& \sum_{e \in P( i,j )} x_e + \sum_{e \in P( k,\ell )} x_e \geq 1 && \parbox{\xxlength}{ for all two pairs $(i,j)$ and $(k,\ell)$ for which $P(i,j) \cap P(k,\ell) = \emptyset$, and $P_2(i,j) \cap P_2(k,\ell) \neq \emptyset$}
\end{align*}
The first family of constraints ensures that at least one edge of the paths between $i$ and $j$, $i$~and $k$, and $j$~and $k$ has to be deleted for each inconsistent triple $i$, $j$ and $k$. The second family of constraints ensures that at least one edge is deleted for every pair of paths between $i$ and $j$, and $k$ and $\ell$ that are disjoint in $T_1$, but for which the corresponding paths in $T_2$ are not disjoint. 

\begin{lemma}
The integrality gap of the integer linear program of  Wu~\cite{Wu09} is at least $\tfrac{16}5$.
\end{lemma}
\begin{proof}
Let $n=2^k$ for some $k$ even. %
We label each internal node in both $T_1$ and $T_2$ with a binary string: the roots get the empty string as label, and given an internal node $u$ its left child gets $u$'s label with a ``0'' appended, and its right child gets $u$'s label with a ``1'' appended.
In $T_1$, the leaves are labeled in the same way as the internal nodes, with a binary string of length $k$. In $T_2$, the binary string is {\em reversed} to give the label of the leaf. For example, the leftmost leaf (of both trees) has label $00\cdots 0$, and the leaf to the right of it has label $0\cdots 01$ in $T_1$, and $10\cdots 0$ in $T_2$.

Consider the internal nodes whose label is a string of length strictly less than $k/2$; there are exactly $2^{k/2}-1=\sqrt{n}-1$ such nodes in each tree. We claim that any component $A$ must cover at least $|A|-1$ of these internal nodes. To see this, consider the set of internal nodes in $V_t$ that have two children in the subtree of $T_t$ induced by $A$ for $t=1,2$. We will call such nodes  \emph{bifurcating}. Observe that there are $2(|A|-1)$ such nodes. Furthermore, since $A$ is \compatible{}, there is a 1-1 mapping $f$ from the bifurcating nodes in $V_1$ to the bifurcating nodes in $V_2$, where, $\lbelow{u}\cap A=\lbelow{f(u)}\cap A$. Now, the label for a bifurcating node $u\in V_1$ is the maximum length prefix that the binary strings for the leaves in $\lbelow{u}\cap A$ have in common, and the label for $f(u)$ is the reverse of the maximum length suffix the leaves in $\lbelow{u}\cap A$ have in common. Hence, at least one of $u$ and $f(u)$'s labels has length less than $k/2$.

The fact that any component $A$ must cover at least $|A|-1$ of the $2\sqrt{n}-2$ internal nodes with labels of length less than $k/2$ implies that any partition that does not overlap must have at least $n-2\sqrt{n}+2$ components. Thus the optimal value of the integer program is at least $n-2\sqrt{n}+1$.

On the other hand, the LP relaxation of the integer program has a feasible solution with objective value $\frac{5}{16}n$: set a value of $\tfrac 14$ on the edges to each leaf in the tree (i.e., from an internal node with a label of length $k-1$ to a node with a label length $k$), and a value of $\tfrac 18$ on all edges between nodes with labels of length $k-2$ to nodes with labels of length $k-1$.
This implies a lower bound of $\lim_{n\to\infty}\frac{n-2\sqrt{n}+1}{\frac{5}{16}n}=\tfrac{16}5$ on the integrality gap.
\end{proof}

As remarked in the introduction, the largest integrality gap for our formulation that we are aware of is $5/4$.
The instance is described in Figure~\ref{fig:intgap}.

\begin{figure}
\begin{center}
    \begin{tikzpicture}[scale=0.9]
\node (T1) at (0,3) {$T_1$};
\node (1) at (0,0) {$1$};
\node (2) at (1,0) {$2$};
\node (3) at (2,0) {$3$};
\node (4) at (3,0) {$4$};
\node (5) at (4,0) {$5$};
\node (6) at (5,0) {$6$};
\node (7) at (6,0) {$7$};
\node (8) at (7,0) {$8$};
\coordinate (12) at (0.5,1);
\coordinate (23) at (1.2,1.5); 
\coordinate (34) at (1.9,1.9); 
\coordinate (45) at (2.6,2.3); 
\coordinate (56) at (3.3,2.6); 
\coordinate (67) at (4,2.85); 
\coordinate (78) at (4.8,3.1); 
\path[c1t] (1) edge (12)
	(2) edge (12)
	(12) edge (23)
	(3) edge (23)
    ;
\path[c3t]
    (4) edge (34)
    (6) edge (56)
    (34) edge (45)
    (45) edge (56)
    (7) edge (67)
    (56) edge (67)
;
\path[c2] (1) edge (12)
	(12) edge (23)
    (23) edge (34)
    (5) edge (45)
    (34) edge (45)
    (45) edge (56)
    (56) edge (67)
    (8) edge (78)
    (67) edge (78)
;
\end{tikzpicture}
\qquad
\begin{tikzpicture}[scale=0.9]
\node (T2) at (0,3) {$T_2$};
\node (1) at (0,0) {$1$};
\node (5) at (1,0) {$5$};
\node (8) at (2,0) {$8$};
\node (2) at (3,0) {$2$};
\node (7) at (4,0) {$7$};
\node (3) at (5,0) {$3$};
\node (6) at (6,0) {$6$};
\node (4) at (7,0) {$4$};
\coordinate (15) at (.5,1);
\coordinate (58) at (1.2,1.5);
\coordinate (27) at (3.5,1);
\coordinate (28) at (1.9,1.9);
\coordinate (36) at (5.5,1);
\coordinate (46) at (6.1,1.5);
\coordinate (r) at (3.7,2.7);
\path[c2t]
(1) edge (15)
(5) edge (15)
(8) edge (58)
(15) edge (58)
;
\path[c3t]
(7) edge (27)
(6) edge (36)
(4) edge (46)
(27) edge (28)
(r) edge (46)
(r) edge (28)
(36) edge (46)
;
\path[c1] 
(1) edge (15)
(15) edge (58)
(2) edge (27)
(27) edge (28)
(58) edge (28)
(36) edge (46)
(3) edge (36)
(r) edge (46)
(r) edge (28)
;
\end{tikzpicture}
\end{center}
\caption{Example with integrality gap of $\tfrac 54$ for the new ILP introduced in Section~\ref{sec:analysis}. An optimal solution to the LP-relaxation is indicated by the colors: the compatible sets corresponding to each of the components $\{1,2,3\}, \{1,5,8\}$ and $\{4,6,7\}$ (indicated by the colors red, blue and green respectively) have an $x$-value of $\tfrac 12$, as well as every singleton leaf set, except leaf $1$; all other $x$-values are $0$. The objective value of this solution is $\tfrac 12(3+7)-1 = 4$. An optimal solution to the ILP has $6$ components, i.e., an objective value of $5$.}\label{fig:intgap}
\end{figure}
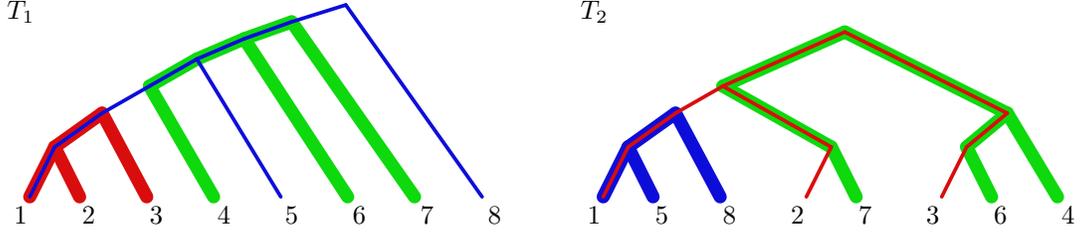

\end{document}